\documentclass[12pt]{article}
\usepackage{amsmath, amssymb, amsthm}
\usepackage{graphicx}
\usepackage{booktabs}
\usepackage{multirow}
\usepackage{url}
\usepackage{geometry}
\usepackage[round]{natbib} 
\usepackage[colorlinks=true, citecolor=blue, linkcolor=blue]{hyperref}
\usepackage{setspace}
\title{Adaptive Long-Run Variance Thresholding for Sparse Covariance Estimation in High-Dimensional Time Series}
\author{Wenhao Zhang and Zhaoxing Gao\thanks{Corresponding author: zhaoxing.gao@uestc.edu.cn (Z. Gao),  School of Mathematical Sciences, University of Electronic Science and Technology of China, Chengdu, 611731 P.R. China.}\\
{\normalsize School of Mathematical Sciences, University of Electronic Science and Technology of China}}
\date{}

\newtheorem{remark}{Remark}
\newtheorem{assumption}{Assumption}
\newtheorem{lemma}{Lemma}
\newtheorem{theorem}{Theorem}

\begin{document}

\maketitle

\begin{abstract}
\noindent Estimating a sparse covariance matrix is a fundamental problem in high-dimensional statistics. However, thresholding methods developed for independent data are generally not directly applicable to high-dimensional time series, where temporal dependence alters the stochastic behavior of sample covariance estimators. This paper studies sparse covariance matrix estimation for high-dimensional time series under weak dependence. We propose a thresholding procedure that incorporates long-run variance into the construction of entry-specific thresholds, thereby adapting to temporal dependence. Under suitable regularity conditions, we show that the proposed estimator is consistent under the spectral norm and attains the optimal convergence rate over a class of sparse covariance matrices. We further establish support recovery consistency for identifying the nonzero entries of the covariance matrix. In addition, we show that universal and adaptive thresholding methods developed for independent data may fail to recover the support consistently in the presence of autocorrelation. Simulation studies demonstrate that the proposed method compares favorably with existing thresholding estimators in terms of both estimation accuracy and support recovery. Applications to gene expression data and stock return data further illustrate its practical usefulness.

\end{abstract}
\newpage

\section{Introduction}

Covariance matrix estimation plays a fundamental role in many fields, including portfolio construction, gene association analysis and statistical inference. In modern statistical applications, large-scale datasets are increasingly common. However, when the dimension \(p\) is comparable to or larger than the sample size \(n\), the sample covariance matrix may fail to be consistent under the spectral norm. To address this issue, researchers often impose structural assumptions on the model. One common assumption is that the multivariate time series follows a low-factor structure. For example, asset returns in finance are frequently modeled as depending on a small number of latent factors; see \citet{1976Ross}, \citet{1989Stock,1998Stock}, and \citet{1993Fama,2014Fama}. Under this framework, the covariance matrix can be estimated through the covariance structure of the common factors; see \citet{2008Fan}. For large-scale data, factor models not only help extract useful common information but also reveal the underlying low-dimensional structure in high-dimensional settings. The key idea is to retain the dominant factors while removing low-impact or negligible noise components, thereby reducing the effective dimensionality. A variety of covariance matrix estimation methods based on factor models have therefore been developed for high-dimensional multivariate time series; see \citet{2015Xu} and \citet{2018Fan}. Alternative approaches for covariance matrix estimation in high-dimensional time series have also been proposed; see \citet{2016Malioutov} and \citet{2024Yang}. 

Another approach is to impose sparsity on the covariance matrix. In real-world high-dimensional data, correlations between variables are often related to their proximity; when the distance between variables is large, such correlations tend to vanish. For example, this phenomenon is commonly observed in financial markets, particularly in stock returns. When attention is restricted to a specific industry, stock returns are often highly correlated within the industry, whereas cross-industry covariances are typically negligible. In portfolio construction, portfolios composed of uncorrelated assets (e.g., gold and technology stocks) tend to produce covariance matrices with many zero entries. More generally, some variables may exhibit little to no correlation. For example, in biological systems, protein–protein interaction networks are often localized, leading to an overall sparse network. Finally, in sensor networks, environmental sensors located far apart often collect data (e.g., temperature and humidity) independently, with correlations occurring primarily among nearby sensors. Therefore, it is essential to develop estimation methods that can accurately identify the zero entries of the covariance matrix.

A basic approach to this problem is to begin with the sample covariance matrix and apply thresholding to each entry in order to remove small estimated values. Sparse covariance matrix estimation has received considerable attention in the literature. \citet{2008Bickel} proposed thresholding-based estimators for the sample covariance matrix $ \boldsymbol{\widehat{\Sigma}}_y $ and established their convergence rates. \citet{2012Cai,2013Cai} established the minimax convergence rates under both the matrix $ L_1 $ norm and the spectral norm. \citet{2011Cai} proposed adaptive thresholding methods and showed that these methods achieve the optimal convergence rate. 

The general thresholding rule was originally proposed by \citet{1994Donoho,1998Donoho} for estimating sparse normal mean vectors in wavelet analysis. A key assumption in this framework is that all noise components share the same variance. Under this assumption, the asymptotic properties of the method can be established. \citet{2011Cai} extended this line of research by showing that sparse covariance matrix estimation is inherently heteroscedastic and that the optimality of general thresholding methods is restricted to a limited parameter space associated with homoscedastic settings. They further introduced a more general parameter space that relaxes the variance constraints while encompassing the original setting. Within this framework, they established the asymptotic optimality of their method and highlighted the suboptimality of general thresholding approaches. They also proposed a data-driven adaptive thresholding procedure. However, most existing methods rely on the assumption of independent and identically distributed data and are therefore not directly applicable to high-dimensional time series.

This paper studies the estimation of sparse covariance matrices in high-dimensional time series. To motivate our approach, we briefly review the adaptive thresholding framework of \citet{2011Cai}, on which our method is built. Interested readers are referred to their original work for further details. \citet{2011Cai} first investigated the estimation of a sparse mean vector in a Gaussian setting with coordinate-specific noise levels. In such a framework, accurate recovery of each mean component requires the corresponding signal to be sufficiently large relative to its own noise level. Under sparsity, universal thresholding can perform well, but it effectively treats a heteroscedastic problem as if all coordinates shared the same noise level. In scenarios where the noise levels across different entries vary significantly, the efficacy of this universal thresholding method is substantially compromised.  A similar issue arises in sparse covariance matrix estimation. Since each column of the covariance matrix can be viewed as a sparse vector, thresholding methods can be applied naturally in this setting. In particular, each sample covariance entry may be regarded as the corresponding population quantity plus a stochastic error whose magnitude depends on its own variance. The main idea of \citet{2011Cai} is to replace a universal threshold by entry-dependent thresholds constructed from estimated entrywise variances. This yields a procedure that adapts to heterogeneous variability across covariance entries. The asymptotic properties of the resulting estimators have been established under suitable regularity conditions, although several issues remain open.

\begin{itemize}
    \item Both the universal thresholding rule and the adaptive thresholding rule are based on the classical central limit theorem and may no longer be valid for time series with autocorrelation.
\end{itemize}

\begin{itemize}
    \item In the presence of autocorrelation, the variance used in adaptive thresholding is no longer the same as the variance appearing in the central limit theorem for the entries of the sample covariance matrix.
\end{itemize}

The primary objective of this paper is to address these issues arising from dependence in high-dimensional time series. We establish conditions under which the central limit theorem holds for high-dimensional time series and incorporate the long-run variance into the thresholding estimator. Specifically, we follow the adaptive thresholding framework of \citet{2011Cai}, but replace the conventional variance estimator with a long-run variance estimator. This modification allows the estimator to account for the cumulative effects of temporal dependence.

Note that \citet{2013Jiang} studied the recovery of the support of nonzero entries in weakly dependent time series with autocorrelation, but did not directly address covariance matrix estimation. Instead, the focus was placed on the correlation matrix, thereby indirectly circumventing the variance issue. In contrast, under weak dependence, this paper directly investigates support recovery for the covariance matrix and establishes strong recovery performance.

We evaluate the performance of the proposed method through finite-sample simulation studies and compare it with existing methods. The results show that, in high-dimensional settings with small sample sizes, the proposed method performs well, whereas existing methods may misclassify zero entries as nonzero. We further apply the proposed method to real data, where it yields covariance matrix estimates with good interpretability and strong empirical performance.

The main contributions of this paper are threefold. First, the proposed thresholding method accommodates temporal dependence and provides accurate covariance estimation for high-dimensional time series. Second, we establish that the proposed estimator achieves the optimal convergence rate and ensures accurate support recovery. Third, we prove that universal and adaptive thresholding methods developed for independent data are generally not support-recovery consistent under autocorrelated dependence.

The rest of the paper is organized as follows. Section \ref{sec2} presents the proposed model and estimation methodology. Section \ref{sec3} establishes the theoretical properties of the proposed estimator. Section \ref{sec4} evaluates the performance of the proposed method using simulated datasets. Section \ref{sec5} evaluates the performance of the proposed method using real datasets. Section \ref{sec6} concludes with a discussion. All technical proofs are provided in the supplement. Throughout the paper, we use the following notation. For a matrix $\mathbf{A} = (a_{i,j})$, $||\mathbf{A}||_2 = \sqrt{\lambda_{\max}(\mathbf{A}^\prime \mathbf{A})}$ denotes the operator norm, where $\lambda_{\max}(\cdot)$ denotes the largest eigenvalue of a matrix and the superscript $\prime$ denotes the transpose of a vector or matrix. Furthermore, $||\mathbf{A}||_{1} = \max_{1 \leq j \leq p} \sum_{i=1}^p |a_{i,j}|$ denotes the $L_1$ norm, $||\mathbf{A}||_\infty = \max_{1 \leq i \leq p} \sum_{j=1}^p |a_{i,j}|$ denotes the $\ell_\infty$ norm, and $||\mathbf{A}||_F = \sqrt{\sum_{i,j=1}^p a_{i,j}^2}$ denotes the Frobenius norm. For two sequences of real numbers ${a_n}$ and ${b_n}$, we write $a_n = O(b_n)$ if there exists a constant $C$ such that $|a_n| \leq C|b_n|$ for all sufficiently large $n$, and we write $a_n = o(b_n)$ if $\lim_{n \to \infty} a_n/b_n = 0$. We also use the notation $a \asymp b$ to mean that $a = O(b)$ and $b = O(a)$.

\section{Methodology}\label{sec2}
\subsection{Setting and Method}
Consider an observable and fourth-order stationary $p \times 1$ time series $\boldsymbol{y}_t = (y_{t1}, y_{t2}, \ldots, y_{tp})^\prime$ with sparse covariance matrix $\boldsymbol{\Sigma}_y$ and mean vector $\boldsymbol{\mu} = (\mu_1, \ldots, \mu_p)^\prime$. Our objective is to estimate $\boldsymbol{\Sigma}_y$. In the independent setting, \citet{2008Bickel} studied the universal thresholding estimator, whereas \citet{2011Cai} developed an adaptive thresholding estimator. For details on thresholding estimation in the independent case, readers are referred to \citet{2011Cai}. We assume that the weakly dependent process ${\boldsymbol{y}_t}$ is $\alpha$-mixing, with mixing coefficients $\alpha_p(k) \to 0$ as $k \to \infty$, where
\begin{equation}  \label{eq1}
    \alpha_p(k)=\sup_i\sup_{A\in \mathcal{F}_{-\infty}^iB\in\mathcal{F}_{i+k}^{+\infty}}|P(A\cap B)-P(A)P(B)|, 
\end{equation}
and $\mathcal{F}_i^j$ denotes the $\sigma$-field generated by ${\boldsymbol{y}_t : i \leq t \leq j}$. Under the moment and mixing conditions specified below, for each fixed pair \((i,j)\),
\[
\frac{
\sqrt n
\left[
n^{-1}\sum_{t=1}^n (y_{it}-\mu_i)(y_{jt}-\mu_j)-\sigma_{ij}
\right]
}{
\sqrt{\theta_{ij}}
}
\Rightarrow N(0,1),
\]
where
\[
\theta_{ij}
=
\lim_{n\to\infty}
\operatorname{Var}
\left(
\frac{1}{\sqrt n}
\sum_{t=1}^n
(y_{it}-\mu_i)(y_{jt}-\mu_j)
\right).
\]

Since the data can be centered by subtracting their sample means, we assume without loss of generality that $\boldsymbol{\mu}=0$. It then follows that
\(\theta_{ij}=\lim_{n\to\infty}\mathrm{Var}\!\left(n^{-1/2}\sum_{t=1}^n y_{it}y_{jt}\right)
\). If $\theta_{ij}$ were known, then the thresholding estimator would be $(\sigma_{ij}^*)_{p \times p}$, where \[\sigma^*_{ij}=s_{\lambda_{ij}}(\hat{\sigma}_{ij}) \quad \text{with}\quad\lambda_{ij}=\delta\sqrt{\frac{\theta_{ij}\log p}{n}},\]and $s_\lambda(z)$ is a thresholding function. In the absence of independence, $\theta_{ij}$ captures the long-run effect of temporal dependence. 

In this paper we consider a class of thresholding functions $s_\lambda(z)$ that satisfy the following conditions:
\begin{enumerate}
    \item \( |s_{\lambda}(z)| \leq |z| \);
    \item \( s_{\lambda}(z) = 0 \) for \( |z| \leq \lambda \);
    \item \( |s_{\lambda}(z) - z| \leq \lambda \).
\end{enumerate}
Condition (1) ensures shrinkage and stabilizes the estimator, condition (2) enforces sparsity, and condition (3) controls the bias. These three conditions are satisfied by the hard-thresholding function $s_\lambda(z) = zI(|z| \geq \lambda)$, the soft-thresholding function $s_\lambda(z) = \mathrm{sgn}(z)(|z|-\lambda)_+$, and the adaptive lasso thresholding function $s_\lambda(z) = z(1-|\lambda/z|^\eta)_+$ with $\eta \geq 1$ and with the convention \(s_\lambda(0)=0\).

\subsection{Estimation}\label{sec22}
Since the estimation of $\boldsymbol{\Sigma}_y$ relies on the long-run variance $\theta_{ij}$, which is unknown in practice, accurate estimation of $\theta_{ij}$ is essential for the proposed inference procedure. Various approaches to long-run variance estimation have been proposed, including the kernel-based estimator of \citet{1991Andrews} and the moving block bootstrap estimator of \citet{2013Lahiri}; see also \citet{1997DenHaan} and \citet{2000Kiefer}. In this paper, we employ a kernel-based estimator for the long-run variance of $z_{ij,t}=y_{it}y_{jt}$:
\begin{equation}  \label{eq3}
    \hat{\theta}_{ij}=\sum_{k=-n+1}^{n-1}\mathcal{K}(\frac{k}{b_n})\hat{\Gamma}_{ij}(k),
\end{equation}
where
\begin{equation}  \label{eq4}
\hat{\Gamma}_{ij}(k)=\begin{cases}
    \frac{1}{n}\sum_{t=k+1}^n[z_{ij,t}-\bar{z}_{ij}][z_{ij,t-k}-\bar{z}_{ij}]\quad \text{if $k\geq0$},\\ \frac{1}{n}\sum_{t=1-k}^n[z_{ij,t+k}-\bar{z}_{ij}][z_{ij,t}-\bar{z}_{ij}]\quad \text{if $k<0$},   
\end{cases}
\end{equation}
with $\bar{z}_{ij}=\frac{1}{n}\sum_{t=1}^n z_{ij,t}$. Here, $K(x)$ is a symmetric function that is continuous at $0$ and satisfies $K(0)=1$, and $b_n$ is a bandwidth parameter that depends on $(i,j)$ and diverges as $n \to \infty$. We omit the subscript $(i,j)$ from $b_n$ for simplicity, since $\hat{\theta}_{ij}$ is a consistent estimator of $\theta_{ij}$ for all $(i,j)$ provided that the bandwidth $b_n$ is chosen at an appropriate rate. Among the many kernel functions that ensure the positive definiteness of long-run variance estimators, \citet{1991Andrews} derived the quadratic spectral kernel
\begin{equation}  \label{eq5}
    \mathcal{K}_{QS}(x)=\frac{25}{12\pi^2x^2}\{\frac{sin(6\pi x/5)}{6\pi x/5}-cos(6\pi x/5) \}
\end{equation}
by minimizing the asymptotic truncated mean squared error of the estimator. \citet{1991Andrews} also studied a data-driven bandwidth selection procedure for the quadratic spectral kernel in Section 6. In the numerical studies in Sections \ref{sec4} and \ref{sec5}, we use kernel functions with clearly specified bandwidth selection procedures. The theoretical results in Lemma 2 apply to general kernel functions.

Using the estimated long-run variance, we construct threshold values and apply thresholding to each entry of the sample covariance matrix $\boldsymbol{\widehat{\Sigma}}_y$, thereby obtaining the threshold estimator $\boldsymbol{\widehat{\Sigma}}_y^*$. The consistency of $\boldsymbol{\widehat{\Sigma}}_y^*$ is established in Theorem \ref{the1} of Section \ref{sec3}.

\section{Theoretical Properties}\label{sec3}
In this section, we show that, under suitable regularity conditions, the estimator $\boldsymbol{\widehat{\Sigma}}_y^*$ is consistent for $\boldsymbol{\Sigma}_y$. As discussed in Section \ref{sec2}, the estimators $\hat{\theta}_{ij}$ of $\theta_{ij}$ play a central role, and our goal is to show that $\hat{\theta}_{ij}$ is a valid estimator of $\theta_{ij}$. Without loss of generality, we assume that $\boldsymbol{\mu} = 0$. Before proceeding, we introduce some notation. The covariance matrix is defined as $\boldsymbol{\Sigma}_y=(\sigma_{ij}),$ where $\sigma_{ij} = E(z_{ij,t})$.

The sample covariance matrix and the thresholding estimators are denoted by \[\boldsymbol{\widehat{\Sigma}}_{y}=(\hat{\sigma}_{ij}),\quad \boldsymbol{\widehat{\Sigma}}^*_y=(\hat{\sigma}^*_{ij}), \quad \boldsymbol{\widehat{\Sigma}}_y^g=(\hat{\sigma}_{ij}^g), \quad \text{and} \quad \boldsymbol{\widehat{\Sigma}}_y^c=(\hat{\sigma}_{ij}^c), \]
where \[\hat{\sigma}_{ij}=\frac{1}{n}\sum_{t=1}^n(y_{it}-\bar{y}_i)(y_{jt}-\bar{y}_j),\quad \hat{\sigma}^*_{ij}=s_{\lambda_{ij}}(\hat{\sigma}_{ij}),\quad \hat{\sigma}^g_{ij}=s_{\lambda^g_{ij}}(\hat{\sigma}_{ij}), \quad\text{and}\quad \hat{\sigma}^c_{ij}=s_{\lambda_{ij}^c}(\hat{\sigma}_{ij}) ,\] with \[\lambda_{ij}=\delta\sqrt{\frac{\hat{\theta}_{ij}\log p}{n}},\quad \lambda_{ij}^g=\delta\sqrt{\frac{\log p}{n}}, \quad\text{and}\quad \lambda_{ij}^c=\delta\sqrt{\frac{\hat{\theta}^c_{ij}\log p}{n}},\] where \[\hat{\theta}_{ij}=\sum_{k=-n+1}^{n-1}\mathcal{K}(\frac{k}{b_n})\hat{\Gamma}_{ij}(k), \quad \hat{\theta}_{ij}^c=\frac{1}{n} \sum_{t=1}^{n} [(y_{it} - \bar{y}_i)(y_{jt} - \bar{y}_j)-\hat{\sigma}_{ij}]^2. \] For simplicity, we omit the subscript $n$, although all of these estimators depend on the sample size $n$.

\subsection{Asymptotic properties when \texorpdfstring{$n \to \infty$}{n to infinity} and \texorpdfstring{$p \to \infty$}{p to infinity}} \label{sec31}
In order to deal with large $p$, we first impose a sparsity condition on the  matrix $\boldsymbol{\Sigma}_y$.
\begin{assumption}
\label{ass1}
For $\boldsymbol{\Sigma}_y=(\sigma_{ij})$, we assume that $\max_{1\leq i \leq p}\sum_{j=1}^p|\sigma_{ij}|^{\iota}\leq s_1$,  for some constant $\iota\in[0,1)$, where $s_1$  may diverge together with $p$. When \(\iota=0\), we use the convention that \(|0|^0=0\), so that the condition counts the number of nonzero entries in each row. We define this class $\boldsymbol{\Sigma}_y$ as 
\[
\begin{aligned}
\mathcal{U}_\iota &:= \mathcal{U}_\iota(s_1) = \left\{ \boldsymbol{\Sigma} : \boldsymbol{\Sigma} \succ 0, \max_i \sum_{j=1}^p |\sigma_{ij}|^\iota \le s_1 \right\}.
\end{aligned}
\]
\end{assumption}

\begin{assumption}
\label{ass2}
There exist $\kappa_1$ and $\kappa_2$ such that $0< \kappa_1\leq \lambda_{\min}(\boldsymbol{\Sigma}_y)\leq \lambda_{\max}(\boldsymbol{\Sigma}_y)\leq \kappa_2$, where \(\kappa_1\) is bounded away from zero and \(\kappa_2\) may depend on \(n\) and \(p\).
\end{assumption}

Unlike the sparsity conditions for the i.i.d. data, in which researchers often put direct structure on the observations, it is not appropriate to do so for dependent data. 

\begin{assumption}
\label{ass3}
There exist constants $C_L>0$, $C_U>0$, $K_1>0$, $K_2>1$ and $\gamma_1\in(0,2]$, such that (i) $\max_{1\leq j\leq p}E\{\exp(K_1|y_{jt}|^{\gamma_1})\}\leq K_2$ and (ii) $0<C_L\leq \theta_{ij} \leq C_U$ for all sufficiently large n.
\end{assumption}

\begin{assumption}
\label{ass4}
For all $k\geq1$, $\sup_p\alpha_p(k)\leq \exp(-K_3k^{\gamma_2})$, where $K_3>0$ and $\gamma_2\in (0,1]$ are constants.
\end{assumption}

 Assumption \ref{ass3}(i) ensures exponentially decaying tails for the statistics considered below, making the analysis applicable to ultra-high-dimensional settings. For sufficient conditions under which Assumption \ref{ass4} holds, see the proof of Lemma 1 in \citet{2019Gao}, which provides a theoretical basis for heteroscedastic autocorrelated models. The restrictions $\gamma_1 \in (0,2]$ and $\gamma_2 \in (0,1]$ are imposed only for the subsequent theoretical analysis. Theorem \ref{the1} below applies when the dimension $p$ satisfies
\begin{equation}  \label{eq7}
    \log p=o(\min(\frac{n}{b_n^2}),n^{\rho_3/(2-\rho_3)}),
\end{equation}
where $\rho_3^{-1}=4\gamma_1^{-1}+\gamma_2^{-1}$.

\begin{assumption}
\label{ass5}
$\mathcal{K}(x)$ is a symmetric kernel function which is continuous at 0 and satisfying $\mathcal{K}(0)=1$, $|\mathcal{K}(x)|\leq1$ for all $x\in R$, and $\int_{-\infty}^{+\infty}|\mathcal{K}(x)|dx\leq C_7<\infty$ for some constant $C_7>0$.
\end{assumption}

Assumption \ref{ass5} includes many commonly used kernel functions, such as the Bartlett, Parzen, Tukey--Hanning, and QS kernels; see \citet{1991Andrews} for details. 

Combining the uniform convergence of the sample covariance entries in Lemma 1 with the uniform consistency of the long-run variance estimators in Lemma 2 in the supplement, we obtain the spectral-norm convergence rate.

\begin{theorem} \label{the1}
    Let Assumptions \ref{ass1}-\ref{ass5} hold. If $\log p=o(\min(\frac{n}{b_n^2}),n^{\rho_3/(2-\rho_3)})$ with a bandwidth $b_n^2=o(n)$, then \[||\boldsymbol{\widehat{\Sigma}}_y^*-\boldsymbol{\Sigma}_y||_2=O_p\{s_1(\frac{\log p}{n})^{\frac{1-\iota}{2}}\},\]
    for a sufficiently large $\delta$, where $s_1$ is defined in Assumption \ref{ass1}.
\end{theorem}

\begin{remark}
    \citet{1991Andrews} systematically investigated the theoretical properties of the kernels and the optimal rates of different kernels are given in Section 6 therein. The optimal order
of $b_n$ is $n^{1/3}$ for the Bartlett kernel and $n^{1/5}$
for the Parzen, Tukey-Hanning and Quadratic spectral kernels, which all satisfy the condition in Theorem \ref{the1} above.
\end{remark}

\subsection{Support recovery} \label{sec32}
Recovery of the support of a sparse covariance matrix is also an important problem in sparse matrix estimation. Define the support of $\boldsymbol{\Sigma}_y = (\sigma_{ij})$ by  $\Psi=\{(i,j):\sigma_{ij}\neq0\}$.

\begin{theorem} \label{the2}
    Let Assumptions \ref{ass1}-\ref{ass5} hold. If $\log p=o(\min(\frac{n}{b_n^2}),n^{\rho_3/(2-\rho_3)})$ with a bandwidth $b_n^2=o(n)$, for a sufficiently large $\delta$, let 
    \begin{equation}  \label{eq8}
        |\sigma_{ij}|\geq(2+\eta)\delta\sqrt{\frac{\theta_{ij}\log p}{n}},
    \end{equation}
    for any $(i,j)\in\Psi$ and some $\eta>0$, we have 
    \begin{equation}  \label{eq9}
        P(\text{supp}(\boldsymbol{\widehat{\Sigma}}_y^*)=\text{supp}(\boldsymbol{\Sigma}_y))\to1.
    \end{equation}
\end{theorem}

Similar support recovery results for correlation matrix estimation under weak dependence were established by \citet{2013Jiang} using adaptive thresholding. Note that Theorem \ref{the2} concerns estimation of the covariance matrix rather than the correlation matrix.

We can also assess support recovery performance using the true positive rate (TPR) and the false positive rate (FPR), defined respectively by 
\begin{equation}  \label{eq10}
    \text{TPR}=\frac{\#\{(i,j):\hat{\sigma}^*_{ij}\neq0 \text{ and }\sigma_{ij}\neq0\text{ and }(i\neq j)\}}{\#\{(i,j):\sigma_{ij}\neq0\text{ and }(i\neq j)\}},
\end{equation}
and
\begin{equation}  \label{eq11}
    \text{FPR}=\frac{\#\{(i,j):\hat{\sigma}^*_{ij}\neq0 \text{ and }\sigma_{ij}=0\text{ and }(i\neq j)\}}{\#\{(i,j):\sigma_{ij}=0\text{ and }(i\neq j)\}}.
\end{equation}
It follows directly from Theorem \ref{the2} that, under condition (\ref{eq8}), \[P(\text{TPR=1)}\to1\quad \text{and} \quad P(\text{FPR=0)}\to1.\]

\subsection{Choice of the parameter \texorpdfstring{$\delta$}{delta}} \label{sec33}
In Sections \ref{sec31} and \ref{sec32}, the theoretical properties are established under the threshold level $\delta \sqrt{\hat{\theta}_{ij}\log p/n}$, where $\delta$ is unknown in practice. When the dimension is large, the performance of the proposed method depends on the choice of the tuning parameter $\delta$. \citet{2008Bickel} proposed selecting the threshold by minimizing the Frobenius norm of the difference between the threshold estimator and the sample covariance matrix computed from independent data. \citet{2011Cai} proposed an adaptive thresholding method along the same lines. For high-dimensional dependent data, random sample splitting may destroy the
temporal ordering of the observations. We therefore use a block cross-validation procedure based on consecutive time blocks. Although we do not provide a theoretical justification for this data-driven procedure, our numerical studies show that it performs reasonably well.

Let \(B_1,\ldots,B_K\) be a partition of \(\{1,\ldots,n\}\) into \(K\) consecutive
blocks with approximately equal lengths. For the \(k\)-th split, the block \(B_k\)
is used as the validation set. To reduce the effect of temporal dependence between
the training and validation samples, we may remove a buffer of length \(\ell_n\)
around \(B_k\). Specifically, define
\[
\mathcal V_k = B_k,
\qquad
\mathcal T_k
=
\{1,\ldots,n\}\setminus
\{t:\operatorname{dist}(t,B_k)\le \ell_n\},
\]
where \(\ell_n\ge 0\). The choice \(\ell_n=0\) gives the ordinary block
cross-validation procedure.

For each candidate value \(\delta\), we compute the proposed thresholding estimator
using only the observations in \(\mathcal T_k\). Denote this estimator by
\[
\widehat\Sigma^{*}_{y,\mathcal T_k}(\delta)
=
\left(
s_{\lambda_{ij,k}(\delta)}
\left(\widehat\sigma_{ij,\mathcal T_k}\right)
\right)_{1\le i,j\le p},
\]
where
\[
\lambda_{ij,k}(\delta)
=
\delta
\sqrt{
\frac{\widehat\theta_{ij,\mathcal T_k}\log p}{|\mathcal T_k|}
}.
\]
Here \(\widehat\sigma_{ij,\mathcal T_k}\) and
\(\widehat\theta_{ij,\mathcal T_k}\) are computed from the training observations
indexed by \(\mathcal T_k\). When \(\mathcal T_k\) consists of several disjoint
time intervals, the lag-\(h\) autocovariances used in
\(\widehat\theta_{ij,\mathcal T_k}\) are computed only from pairs of observations
that both belong to \(\mathcal T_k\).

The validation covariance matrix is computed from the observations in \(B_k\):
\[
\boldsymbol{\widehat\Sigma}_{y,\mathcal V_k}
=
\frac{1}{|\mathcal V_k|}
\sum_{t\in \mathcal V_k}
(y_t-\bar y_{\mathcal V_k})(y_t-\bar y_{\mathcal V_k})',
\]
where $\bar y_{\mathcal V_k} = \frac{1}{|\mathcal V_k|} \sum_{t\in \mathcal V_k}y_t.$
The block cross-validation loss is defined by
\[
\widehat R(\delta)
=
\frac{1}{K}
\sum_{k=1}^{K}
\left\|
\boldsymbol{\widehat\Sigma}^{*}_{y,\mathcal T_k}(\delta)
-
\boldsymbol{\widehat\Sigma}_{y,\mathcal V_k}
\right\|_F^2 .
\]
Let \(a_j=j/M\), \(0\le j\le 4M\), be grid points in \([0,4]\). We select
\[
\widehat\delta
=
\arg\min_{\delta\in\{a_0,\ldots,a_{4M}\}}
\widehat R(\delta).
\]

\subsection{Comparison with other thresholding methods}
Section \ref{sec31} established the convergence rate of the proposed estimator. In this section, we study the support recovery performance of these thresholding estimators under a relatively simple dependence structure.

For the comparison results below, we consider the diagonal $\mathrm{VAR}(1)$ Gaussian process
\[
\boldsymbol{y}_t=\boldsymbol\Phi y_{t-1}+\boldsymbol\varepsilon_t,
\qquad
\boldsymbol\Phi=\operatorname{diag}(\underbrace{0,\ldots,0}_{s_0},
\underbrace{\rho,\ldots,\rho}_{p-s_0}),
\]
where \(0<\rho<1\) is a fixed constant and \(s_0=\lfloor s_1/4\rfloor\).
Let
\[
\boldsymbol\Sigma_y=
\begin{pmatrix}
\mathbf{A}_1 & 0\\
0 & v_* I_{p-s_0}
\end{pmatrix},
\]
where
\[
\mathbf{A}_1=(1-a_n)I_{s_0}+a_n\mathbf 1\mathbf 1',
\qquad
a_n=c_a\sqrt{\frac{\log p}{n}},
\]
with \(c_a>0\) and $v_*=1$ are fixed constants.
The innovation covariance matrix is chosen as $\boldsymbol{\Sigma}_\varepsilon=\boldsymbol{\Sigma}_y-\Phi\boldsymbol{\Sigma}_y\Phi$,
so that \(\boldsymbol\Sigma_y\) is the stationary covariance matrix of the process.

Under this construction, the off-diagonal entries in the first block satisfy $\sigma_{ij}=a_n, 1\le i\neq j\le s_0$, whereas all off-diagonal entries in the second block are zero: $\sigma_{ij}=0,s_0+1\le i\neq j\le p$. Moreover, for \(s_0+1\le i\neq j\le p\),
\[
\theta_{ij}=\frac{1+\rho^2}{1-\rho^2}=: \theta_0,
\]
which is a positive constant depending on \(\rho\). We also choose the constants \(\rho\) and \(c_a\) so that the above construction is
compatible with Theorem \ref{the2} for the proposed estimator.
Let \(\delta_*>0\) denote the fixed threshold constant used in the proposed estimator.
Since, for \(1\le i\neq j\le s_0\),
\[
\sigma_{ij}=a_n=c_a\sqrt{\frac{\log p}{n}},
\qquad
\theta_{ij}=1+a_n^2,
\]
the signal strength condition in Theorem \ref{the2},
\[
|\sigma_{ij}|\ge (2+\eta)\delta_*\sqrt{\theta_{ij}\frac{\log p}{n}},
\]
is satisfied for all sufficiently large \(n\) provided that $c_a>(2+\eta)\delta_*$.
Therefore, under this construction, the proposed estimator still satisfies $P(\mathrm{TPR}=1)\to1,$ and $P(\mathrm{FPR}=0)\to1$,
whereas the competing thresholding estimators considered below may fail to recover
the support consistently.
\begin{theorem} \label{the3}
Assume that for some $\xi>0$, $n^\xi\le p$, $\log p= o(\min(\frac{n}{b_n^2}),n^{\rho_3/(2-\rho_3)})$ and for some $\gamma\in(0,1),8\leq s_1=O((\log p)^{\gamma})$. Under the above construction, there exist constants \(\rho\in(0,1)\) and \(c_a>0\),
independent of \(n\), \(p\), and \(s_1\), such that
\[
\inf_{\delta>0}\sup_{\boldsymbol{\Sigma}_y\in U_{\iota}(s_1)}
P\Bigl(
\mathrm{TPR}(\widehat{\boldsymbol{\Sigma}}_y^{g})<1
\ \textrm{or}\
\mathrm{FPR}(\widehat{\boldsymbol{\Sigma}}_y^{g})>0
\Bigr)\to1.
\]
Equivalently,
\[
\inf_{\delta>0}\sup_{\boldsymbol{\Sigma}_y\in U_{\iota}(s_1)}
P\Bigl(
\operatorname{supp}(\widehat{\boldsymbol{\Sigma}}_y^{g})
\neq
\operatorname{supp}(\boldsymbol{\Sigma}_y)
\Bigr)\to1.
\]
Therefore, the universal thresholding estimator is not support recovery consistent
under the dependent setting considered here.
\end{theorem}
Theorem \ref{the3} shows that, under temporal dependence, the universal thresholding estimator may fail to recover the support consistently even when both the covariance entries and the long-run variances are uniformly bounded. A single universal threshold cannot simultaneously avoid false inclusions in the persistent zero block and preserve weak but nonzero signals in the first sparse block.

We next consider the adaptive thresholding estimator of \citet{2011Cai}. Intuitively, for autocorrelated time series, the estimator $\hat{\theta}_{ij}^c$ in \citet{2011Cai} captures only part of $\theta_{ij}$. Under autocorrelation, the contemporaneous variance $\gamma_{z_{ij}}(0)=\mathrm{Var}(z_{ij,t})$ does not coincide with the long-run variance $\theta_{ij}$ governing the fluctuation of the sample covariance entry. As a result, the adaptive thresholding rule of \citet{2011Cai}, which is calibrated to $\gamma_{z_{ij}}(0)$, may fail to recover the support consistently.  For brevity, we focus on the Gaussian case. Similar to the setting considered in Theorem \ref{the3}, we show that, in the presence of autocorrelation, universal and adaptive thresholding methods developed for independent data may fail to achieve consistent support recovery, in the sense that the true positive rate may fail to converge to one or the false positive rate may fail to converge to zero.
\begin{theorem}  \label{the4}
Assume that for some $\xi>0$, $n^\xi\le p$, $\log p= o(\min(\frac{n}{b_n^2}),n^{\rho_3/(2-\rho_3)})$ and for some $\gamma\in(0,1),8\leq s_1=O((\log p)^{\gamma})$.
Under the same construction as in Theorem \ref{the3}, there exist constants \(\rho\in(0,1)\)
and \(c_a>0\), independent of \(n\), \(p\), and \(s_1\), such that
\[
\inf_{\delta>0}\sup_{\Sigma_y\in U_{\iota}(s_1)}
P\Bigl(
\mathrm{TPR}(\widehat{\boldsymbol{\Sigma}}_y^{c})<1
\ \textrm{or}\
\mathrm{FPR}(\widehat{\boldsymbol{\Sigma}}_y^{c})>0
\Bigr)\to1.
\]
Equivalently,
\[
\inf_{\delta>0}\sup_{\boldsymbol{\Sigma}_y\in U_{\iota}(s_1)}
P\Bigl(
\operatorname{supp}(\widehat{\boldsymbol{\Sigma}}_y^{c})
\neq
\operatorname{supp}(\boldsymbol{\Sigma}_y)
\Bigr)\to1.
\]
Therefore, the adaptive thresholding estimator of \citet{2011Cai} is not support
recovery consistent under the dependent setting considered here.
\end{theorem}

Theorem \ref{the4} shows that the adaptive thresholding estimator of \citet{2011Cai} may
also fail to recover the support consistently under temporal dependence. Although the
entrywise variance estimator used there is suitable for contemporaneous fluctuations, it
does not correctly account for the long-run variance that determines the stochastic
behavior of sample covariance entries in dependent data. This issue already arises in relatively simple autocorrelated settings and becomes even more pronounced in more complex ones. For such more general settings, however, this paper does not provide a theoretical justification, and we leave this problem for future research.

\section{Simulation Study}\label{sec4}
In this section, we evaluate the numerical performance of the proposed thresholding estimator through Monte Carlo simulations. For the proposed estimator, the regularization parameter \(\delta\) is selected by the block cross-validation procedure described in Section \ref{sec33}. The proposed method is compared with the universal thresholding estimator and the adaptive thresholding estimator, for which the tuning parameter is selected by ordinary five-fold cross-validation.

\subsection{Simulation setting and performance evaluation}\label{sec4.1}
Following \citet{2011Cai}, we consider two types of covariance matrices to study the numerical performance of the proposed method:
\begin{itemize}
    \item $\text{Model 1}$ (sparse matrix without ordering): $\boldsymbol{\Sigma}_{y} = \operatorname{diag}(\mathbf{A_{1}},\ \mathbf{A_{2}})$,  
  where $\mathbf{A_{2}} = 4\mathbf{I}_{p/2 \times p/2}$, 
  $\mathbf{A_{1}} = \mathbf{B} + \epsilon \mathbf{I}_{p/2 \times p/2}$, 
  $\mathbf{B} = (b_{ij})_{p/2 \times p/2}$ with independent 
  $b_{ij} = \operatorname{unif}(0.3,\,0.8) \times \operatorname{Ber}(1,\,0.2)$ for $i\geq j$ and $b_{ji}=b_{ij}$. Here $\operatorname{unif}(0.3,\,0.8)$ is a random variable taking value uniformly in $[0.3,\,0.8]$; 
  $\operatorname{Ber}(1,\,0.2)$ is a Bernoulli random variable that takes value $1$ with probability $0.2$ and value $0$ with probability $0.8$, 
  and $\epsilon = \max(-\lambda_{\min}(\mathbf{B}),\,0) + 0.01$ to ensure that $\mathbf{A}_1$ is positive definite.
\end{itemize}
\begin{itemize}
    \item $\text{Model 2}$ (banded matrix with ordering): $\boldsymbol{\Sigma}_{y} = \operatorname{diag}(\mathbf{A}_1,\ \mathbf{A}_2)$, 
  where $\mathbf{A}_1 = (\sigma_{ij})_{1 \leq i,j \leq p/2}$, 
  $\sigma_{ij} = \bigl(1 - \frac{|i-j|}{10}\bigr)_{+}$, 
  $\mathbf{A}_2 = 4\mathbf{I}_{p/2 \times p/2}$. 
  $\mathbf{\Sigma}_{y}$ is a two-block diagonal matrix, 
  $\mathbf{A_{1}}$ is a banded and sparse covariance matrix, 
  and $\mathbf{A_{2}}$ is a diagonal matrix with $4$ along the diagonal.
\end{itemize}
Under each model, we consider a coefficient matrix $\mathbf{C}=\text{diag}(\mathbf{D}_1,\mathbf{D}_2)$, where $\mathbf{D}_1=0.5\mathbf{I}_{p/2\times p/2}$ and $\mathbf{D}_2=0.8\mathbf{I}_{p/2\times p/2}$. A $\mathrm{VAR}(1)$ process with mean vector $\boldsymbol{\mu}=0$ and covariance matrix $\boldsymbol{\Sigma}_y$ is generated for $p=100,150,200$ and $n=100,500,800$. The innovation covariance matrix is set to $\boldsymbol\Sigma_\varepsilon=\boldsymbol\Sigma_y-\mathbf{C}\boldsymbol\Sigma_y\mathbf{C}'$,
so that the stationary covariance matrix of the $\mathrm{VAR}(1)$ process is \(\boldsymbol\Sigma_y\). Each configuration is replicated $100$ times. We compare the numerical performance of the proposed estimator $\boldsymbol{\widehat{\Sigma}}_y^*$ with that of the universal thresholding estimator $\boldsymbol{\widehat{\Sigma}}_y^g$ of \citet{2009Rothman} and the adaptive thresholding estimator $\boldsymbol{\widehat{\Sigma}}_y^c$ of \citet{2011Cai}, using both hard thresholding and adaptive lasso thresholding. The threshold level $\delta$ in both $\boldsymbol{\widehat{\Sigma}}_y^g$ and $\boldsymbol{\widehat{\Sigma}}_y^c$ is selected by five-fold cross-validation. Estimation error is measured under the spectral norm. In all simulations, the long-run variance estimator was implemented using the 
default \text{lrvar} function in the R package \text{sandwich}. This corresponds 
to the Andrews HAC estimator with the quadratic spectral kernel and automatic 
bandwidth selection. The block cross-validation procedure used \(K=5\) consecutive 
blocks, no buffer region \((\ell_n=0)\), and a grid size \(M=10\), so that 
\(\delta\in\{0,0.1,\ldots,4.0\}\).

To remain consistent with the theory in Section \ref{sec3}, we define the loss function for an estimator $\boldsymbol{\widehat{\Sigma}}_y$ by
\begin{equation}  \label{eq16}
    L(\boldsymbol{\widehat{\Sigma}}_y,\boldsymbol{\Sigma}_y)=||\boldsymbol{\widehat{\Sigma}}_y-\boldsymbol{\Sigma}_y||_2.
\end{equation}

The ability to recover sparsity is evaluated by the true positive rate (TPR) together with the false positive rate (FPR), as defined in (\ref{eq10}) and (\ref{eq11}).

\subsection{Summary of numerical results}

\begin{table}
\centering
\caption{Comparison of average matrix losses for Model 1  over 100 replications. The standard errors are given in parentheses.}
\label{tab:model1_matrix_loss_comparison}
\begin{tabular}{ccccccc}
\toprule
 & \multicolumn{3}{c}{hard} & \multicolumn{3}{c}{adaptive lasso} \\
\cmidrule(lr){2-4} \cmidrule(lr){5-7}
$p$ & $\boldsymbol{\widehat{\Sigma}}^{g}_y$ & $\boldsymbol{\widehat{\Sigma}}^*_y$ & $\boldsymbol{\widehat{\Sigma}}^c_y$ & $\boldsymbol{\widehat{\Sigma}}^{g}_y$ & $\boldsymbol{\widehat{\Sigma}}^*_y$ & $\boldsymbol{\widehat{\Sigma}}^c_y$ \\
\midrule
\multicolumn{7}{c}{$n=100$}
\\
 100 & 7.35(1.74) & \textbf{5.54(0.26)} & 6.39(0.86) & 6.17(0.54) & \textbf{5.36(0.26)} & 5.93(0.17) \\
 150 & 9.66(1.20) & \textbf{8.11(0.30)} & 9.05(0.35) & 8.97(0.25) & \textbf{7.81(0.33)} & 8.86(0.17) \\
 200 & 12.00(1.20) & \textbf{10.52(0.40)} & 11.55(0.22) & 11.42(0.22) & \textbf{10.16(0.39)} & 11.41(0.17) \\ 
\midrule
\multicolumn{7}{c}{$n=500$}
\\
 100 & 4.28(0.36) & \textbf{2.72(0.17)} & 3.65(0.27) & 4.04(0.32) & \textbf{2.63(0.18)} & 3.45(0.25) \\
 150 & 5.77(0.39) & \textbf{4.06(0.23)} & 5.37(0.32) & 5.55(0.36) & \textbf{3.99(0.26)} & 5.15(0.29) \\
 200 & 7.18(0.45) & \textbf{5.48(0.26)} & 7.18(0.36) & 6.95(0.40) & \textbf{5.42(0.28)} & 6.92(0.33) \\ 
\midrule
\multicolumn{7}{c}{$n=800$}
\\
 100 & 3.13(0.23) & \textbf{2.01(0.11)} & 2.64(0.18) & 2.99(0.21) & \textbf{1.96(0.14)} & 2.53(0.16) \\
 150 & 4.20(0.32) & \textbf{2.96(0.18)} & 3.89(0.27) & 4.07(0.28) & \textbf{2.93(0.18)} & 3.78(0.24) \\
 200 & 5.15(0.29) & \textbf{4.01(0.18)} & 5.19(0.25) & 5.07(0.27) & \textbf{4.02(0.21)} & 5.09(0.23) \\
\bottomrule
\end{tabular}
\end{table}

$\quad$ Table \ref{tab:model1_matrix_loss_comparison} summarizes results of average matrix loss for Model 1. According to Table \ref{tab:model1_matrix_loss_comparison}, under both the hard-thresholding rule and the adaptive lasso rule, the proposed estimator consistently outperforms the universal and adaptive thresholding estimators across all parameter settings. The estimation error decreases as the sample size $n$ increases, whereas for fixed $n$, it increases with the dimension $p$ because of the growing number of nonzero entries. Among the thresholding rules, adaptive lasso generally performs better than hard thresholding, although the improvement is modest.

\begin{table}
\centering
\caption{Comparison of average matrix losses for  Model 2 over 100 replications. The standard errors are given in parentheses.}
\label{tab:model2_matrix_loss_comparison}
\begin{tabular}{lccccccc}
\toprule
& \multicolumn{3}{c}{hard} & & \multicolumn{3}{c}{adaptive lasso} \\
\cmidrule{2-4} \cmidrule{6-8}
$p$ & $\boldsymbol{\widehat{\Sigma}}_y^g$ & $\boldsymbol{\widehat{\Sigma}}_y^*$ & $\boldsymbol{\widehat{\Sigma}}_y^c$ & & $\boldsymbol{\widehat{\Sigma}}_y^g$ & $\boldsymbol{\widehat{\Sigma}}_y^*$ & $\boldsymbol{\widehat{\Sigma}}_y^c$ \\
\midrule
\multicolumn{8}{c}{$n=100$} \\
100 & 8.93(0.62) & \textbf{3.80(0.97)} & 5.13(1.16) & & 8.74(0.08) & \textbf{3.74(0.78)} & 5.16(0.76) \\
150 & 9.56(1.47) & \textbf{4.24(0.72)} & 6.10(1.17) & & 8.97(0.40) & \textbf{4.21(0.65)} & 5.84(0.70) \\
200 & 10.38(2.15) & \textbf{4.72(0.96)} & 6.64(1.19) & & 9.15(0.89) & \textbf{4.65(0.87)} & 6.36(0.73) \\
\midrule
\multicolumn{8}{c}{$n=500$} \\
100 & 4.63(0.41) & \textbf{1.57(0.36)} & 1.99(0.44) & & 4.59(0.37) & \textbf{1.52(0.33)} & 1.73(0.29) \\
150 & 5.60(0.52) & \textbf{1.70(0.35)} & 2.21(0.50) & & 5.41(0.46) & \textbf{1.64(0.27)} & 1.94(0.29) \\
200 & 6.25(0.46) & \textbf{1.76(0.28)} & 2.33(0.40) & & 6.01(0.39) & \textbf{1.74(0.27)} & 2.07(0.24) \\
\midrule
\multicolumn{8}{c}{$n=800$} \\
100 & 3.26(0.30) & \textbf{1.24(0.31)} & 1.63(0.33) & & 3.24(0.29) & \textbf{1.19(0.28)} & 1.31(0.27) \\
150 & 3.92(0.29) & \textbf{1.32(0.25)} & 1.72(0.33) & & 3.87(0.28) & \textbf{1.26(0.21)} & 1.42(0.20) \\
200 & 4.28(0.34) & \textbf{1.37(0.25)} & 1.79(0.30) & & 4.25(0.31) & \textbf{1.35(0.23)} & 1.52(0.21) \\
\bottomrule
\end{tabular}
\end{table}
Table \ref{tab:model2_matrix_loss_comparison} gives results of average matrix loss for Model 2. According to Table \ref{tab:model2_matrix_loss_comparison}, the proposed method based on long-run variance continues to outperform the competing thresholding methods across all parameter settings. The adaptive lasso rule generally performs better than the hard-thresholding rule, although this advantage diminishes as the sample size increases. Compared with Model 1, the growth in estimation error with respect to dimensionality is more moderate, owing to the bounded number of nonzero entries.

Tables \ref{tab:model1_tpr_fpr_comparison} and \ref{tab:model2_tpr_fpr_comparison} report the TPR and FPR for support recovery based on the off-diagonal entries. 

\begin{table}
\centering
\caption{Comparison of average TPR for Model 1  over 100 replications. The FPRs are given in parentheses.}
\label{tab:model1_tpr_fpr_comparison}
\begin{tabular}{ccccccc}
\toprule
 & \multicolumn{3}{c}{hard} & \multicolumn{3}{c}{adaptive lasso} \\
\cmidrule(lr){2-4} \cmidrule(lr){5-7}
$p$ & $\boldsymbol{\widehat{\Sigma}}^{g}_y$ & $\boldsymbol{\widehat{\Sigma}}^*_y$ & $\boldsymbol{\widehat{\Sigma}}^c_y$ & $\boldsymbol{\widehat{\Sigma}}^{g}_y$ & $\boldsymbol{\widehat{\Sigma}}^*_y$ & $\boldsymbol{\widehat{\Sigma}}^c_y$ \\
\midrule
\multicolumn{7}{c}{$n=100$}
\\
 100 & 0.00(0.00) & \textbf{0.08(0.01)} & 0.00(0.00) & 0.00(0.00) & \textbf{0.19(0.04)} & 0.03(0.02) \\
 150 & 0.00(0.00) & \textbf{0.08(0.01)} & 0.00(0.00) & 0.01(0.01) & \textbf{0.20(0.05)} & 0.03(0.02) \\
 200 & 0.00(0.00) & \textbf{0.07(0.01)} & 0.00(0.00) & 0.02(0.01) & \textbf{0.17(0.05)} & 0.02(0.02) \\
\midrule
\multicolumn{7}{c}{$n=500$}
\\
 100 & 0.24(0.02) & \textbf{0.56(0.01)} & 0.33(0.02) & 0.46(0.05) & \textbf{0.74(0.05)} & 0.56(0.05) \\
 150 & 0.30(0.02) & \textbf{0.54(0.02)} & 0.33(0.03) & 0.51(0.06) & \textbf{0.71(0.08)} & 0.55(0.07) \\
 200 & 0.30(0.03) & \textbf{0.49(0.03)} & 0.28(0.03) & 0.49(0.07) & \textbf{0.65(0.09)} & 0.49(0.09) \\
\midrule
\multicolumn{7}{c}{$n=800$}
\\
 100 & 0.44(0.02) & \textbf{0.71(0.01)} & 0.52(0.01) & 0.66(0.05) & \textbf{0.85(0.04)} & 0.73(0.05) \\
 150 & 0.49(0.02) & \textbf{0.69(0.01)} & 0.51(0.02) & 0.69(0.06) & \textbf{0.83(0.07)} & 0.72(0.07) \\
 200 & 0.49(0.03) & \textbf{0.64(0.02)} & 0.47(0.03) & 0.68(0.08) & \textbf{0.79(0.09)} & 0.67(0.09) \\
\bottomrule
\end{tabular}
\end{table}
The results in Table \ref{tab:model1_tpr_fpr_comparison} indicate that $\boldsymbol{\widehat{\Sigma}}_y^*$ provides a favorable balance between TPR and FPR compared with $\boldsymbol{\widehat{\Sigma}}_y^c$ and $\boldsymbol{\widehat{\Sigma}}_y^g$. The adaptive lasso rule yields higher TPR and FPR than the hard-thresholding rule because it typically selects smaller regularization parameters, which leads to less aggressive thresholding. However, the FPR under the adaptive lasso rule is generally higher than that under the hard-thresholding rule. Therefore, we can select different threshold rules based on varying requirements for the TPR and FPR. 

\begin{table}
\centering
\caption{Comparison of average TPR for Model 2  over 100 replications. The FPRs are given in parentheses.}
\label{tab:model2_tpr_fpr_comparison}
\begin{tabular}{lccccccc}
\toprule
& \multicolumn{3}{c}{hard} & & \multicolumn{3}{c}{adaptive lasso} \\
\cmidrule{2-4} \cmidrule{6-8}
$p$ & $\boldsymbol{\widehat{\Sigma}}_y^g$ & $\boldsymbol{\widehat{\Sigma}}_y^*$ & $\boldsymbol{\widehat{\Sigma}}_y^c$ & & $\boldsymbol{\widehat{\Sigma}}_y^g$ & $\boldsymbol{\widehat{\Sigma}}_y^*$ & $\boldsymbol{\widehat{\Sigma}}_y^c$ \\
\midrule
\multicolumn{8}{c}{$n=100$} \\
100 & 0.00(0.00) & \textbf{0.60(0.00)} & 0.39(0.00) & & 0.00(0.01) & \textbf{0.70(0.02)} & 0.55(0.01) \\
150 & 0.00(0.00) & \textbf{0.58(0.00)} & 0.33(0.00) & & 0.00(0.01) & \textbf{0.69(0.01)} & 0.50(0.01) \\
200 & 0.00(0.00) & \textbf{0.55(0.00)} & 0.29(0.00) & & 0.00(0.00) & \textbf{0.70(0.01)} & 0.47(0.00) \\
\midrule
\multicolumn{8}{c}{$n=500$} \\
100 & 0.36(0.02) & \textbf{0.85(0.00)} & 0.77(0.00) & & 0.53(0.05) & \textbf{0.88(0.01)} & 0.81(0.00) \\
150 & 0.30(0.02) & \textbf{0.84(0.00)} & 0.75(0.00) & & 0.47(0.04) & \textbf{0.87(0.00)} & 0.79(0.00) \\
200 & 0.26(0.01) & \textbf{0.83(0.00)} & 0.74(0.00) & & 0.43(0.03) & \textbf{0.86(0.00)} & 0.79(0.00) \\
\midrule
\multicolumn{8}{c}{$n=800$} \\
100 & 0.47(0.02) & \textbf{0.90(0.00)} & 0.83(0.00) & & 0.63(0.05) & \textbf{0.91(0.00)} & 0.85(0.00) \\
150 & 0.42(0.01) & \textbf{0.88(0.00)} & 0.81(0.00) & & 0.58(0.04) & \textbf{0.90(0.00)} & 0.84(0.00) \\
200 & 0.39(0.01) & \textbf{0.88(0.00)} & 0.81(0.00) & & 0.56(0.03) & \textbf{0.90(0.00)} & 0.84(0.00) \\
\bottomrule
\end{tabular}
\end{table}
Similar patterns are observed in Table \ref{tab:model2_tpr_fpr_comparison}, although the TPR is generally higher than that in Table \ref{tab:model1_tpr_fpr_comparison}. This is due to the structure of Model 2, in which each row contains a fixed number of nonzero entries. Consequently, the estimation problem is simpler, leading to improved support recovery and a higher TPR. Meanwhile, the larger proportion of zero entries encourages more conservative thresholding, resulting in a lower FPR.

To provide a clearer illustration of support recovery, we plot heatmaps of the estimated supports based on 100 replications for $p = 150$ with $n = 800$ and $p=100$ with $n=100$ in Figures \ref{fig:model1_n800}, \ref{fig:model2_n800}, \ref{fig:model1_n100} and \ref{fig:model2_n100}. The adaptive lasso thresholding rule typically selects smaller regularization parameters, resulting in a higher TPR but also a higher FPR than hard thresholding. By contrast, the proposed method based on long-run variance produces recovery patterns that more accurately reflect the true covariance structure.
\begin{figure}[htbp]
    \centering
    
        \includegraphics[width=\textwidth]{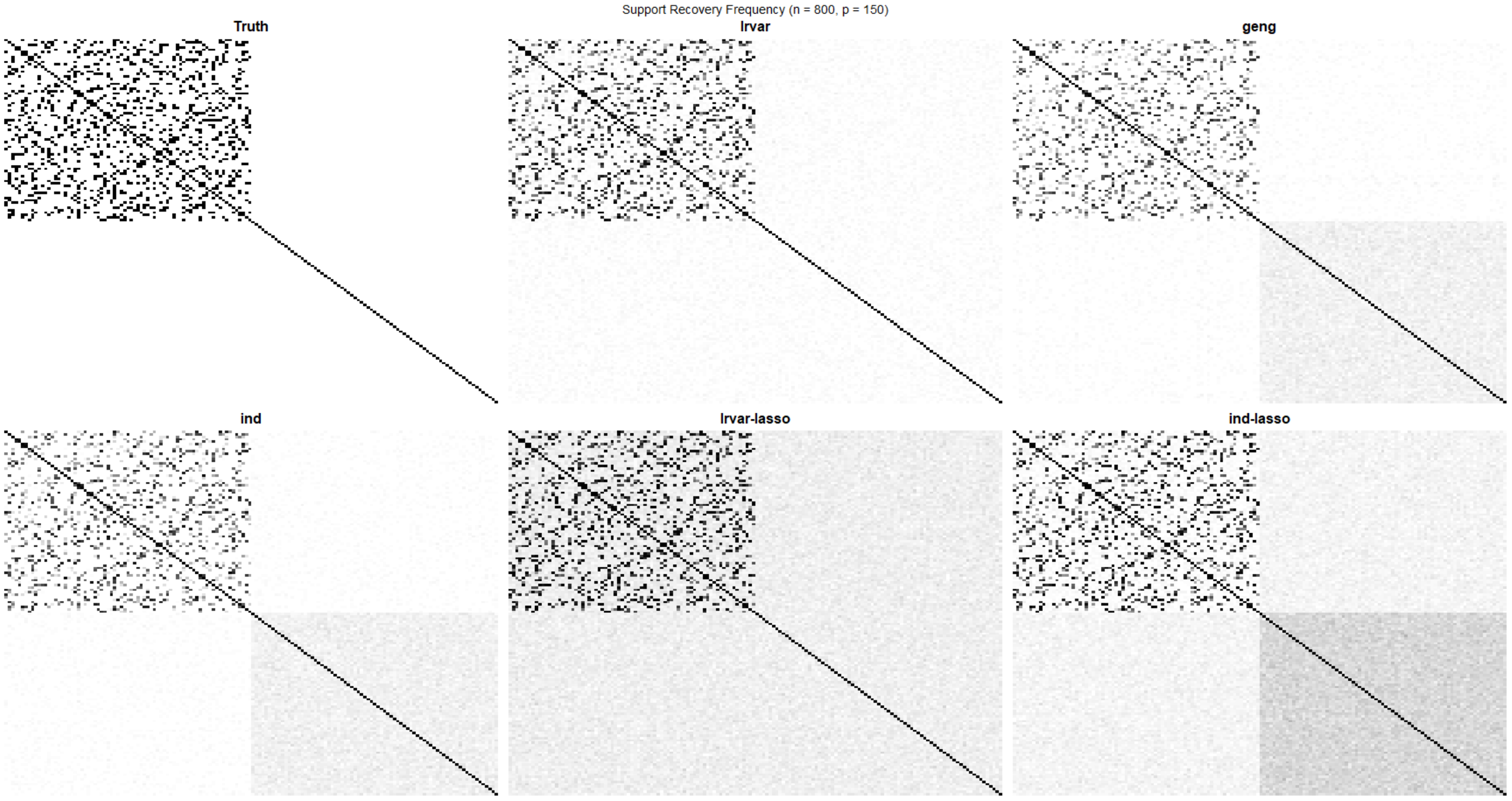}
    
    \caption{Heatmaps of the frequency of the nonzero entries identified for each entry of the Model 1 covariance matrix (when $p = 150$ and $n=800$) out of 100 replications. Black indicates 100 nonzeros identified out of 100 runs; white indicates 0/100.}
    \label{fig:model1_n800}
\end{figure}

Figure \ref{fig:model1_n800} shows that the proposed estimator $\boldsymbol{\widehat{\Sigma}}_y^*$ achieves better support recovery than the competing methods. Adaptive lasso thresholding yields better recovery of nonzero entries, whereas hard thresholding is more effective in identifying zero entries. In contrast, the universal and adaptive thresholding estimators $\boldsymbol{\widehat{\Sigma}}_y^c$ and $\boldsymbol{\widehat{\Sigma}}_y^g$ exhibit substantial misclassification of zero entries in the diagonal block $\mathbf{A}_2$. These findings are consistent with the theoretical results established in this paper.

\begin{figure}[htbp]
    \centering
    
        \includegraphics[width=\textwidth]{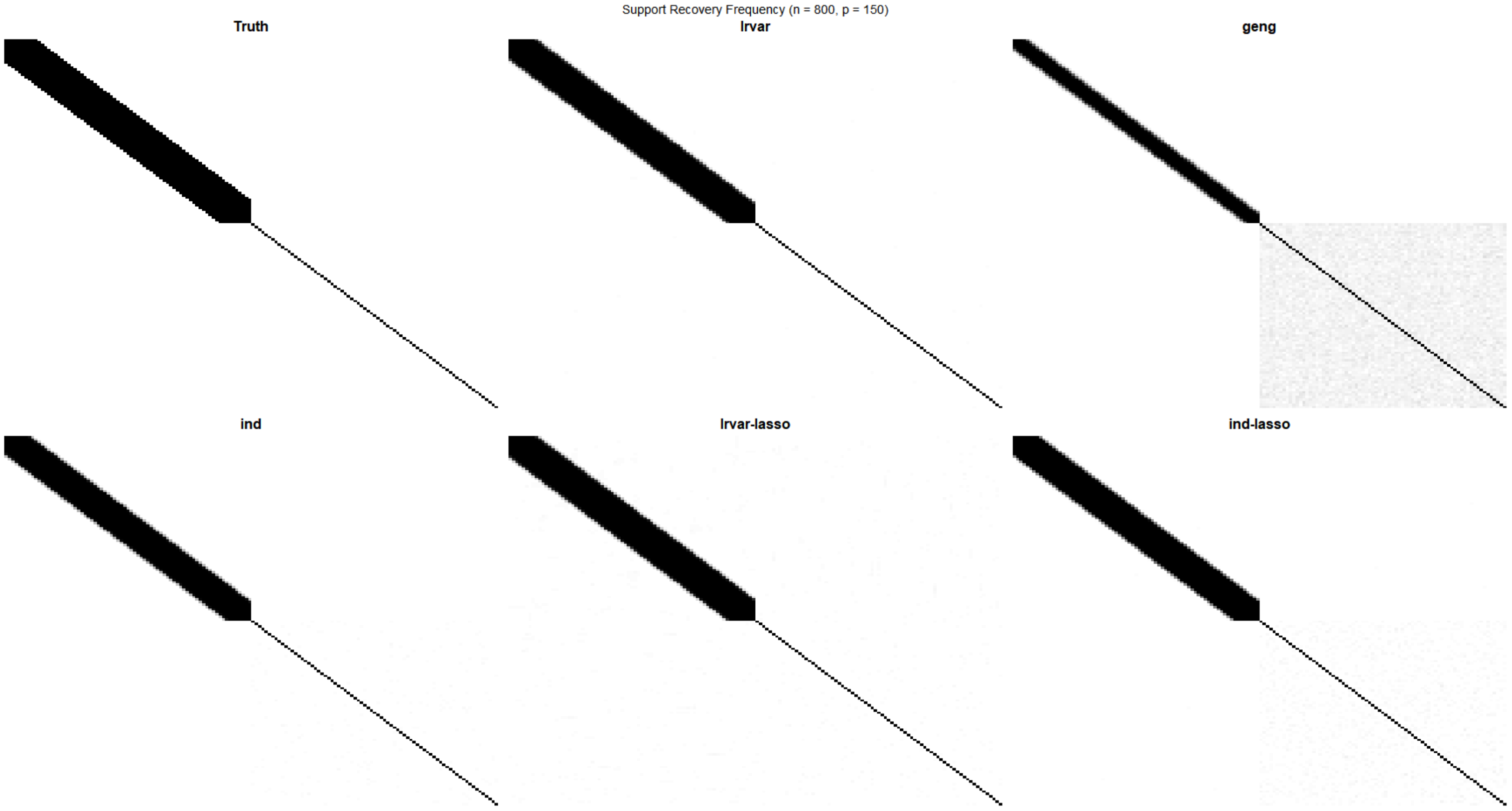}
    
    \caption{Heatmaps of the frequency of the nonzero entries identified for each entry of the Model 2 covariance matrix (when $p = 150$ and $n=800$) out of 100 replications. Black indicates 100 nonzeros identified out of 100 runs; white indicates 0/100.}
    \label{fig:model2_n800}
\end{figure}

Figure \ref{fig:model2_n800} shows that the proposed method based on long-run variance consistently outperforms the competing thresholding estimators. Adaptive lasso rule improves the recovery of zero entries but sacrifices accuracy in identifying nonzero entries. This phenomenon, which is consistent with the results in Table \ref{tab:model2_tpr_fpr_comparison}, arises because a low proportion of nonzero entries leads to the selection of larger regularization parameters, which favor sparsity and improve recovery of zero entries at the expense of detecting nonzero ones.

\begin{figure}[htbp]
    \centering
    
        \includegraphics[width=\textwidth,height=0.4\textheight]{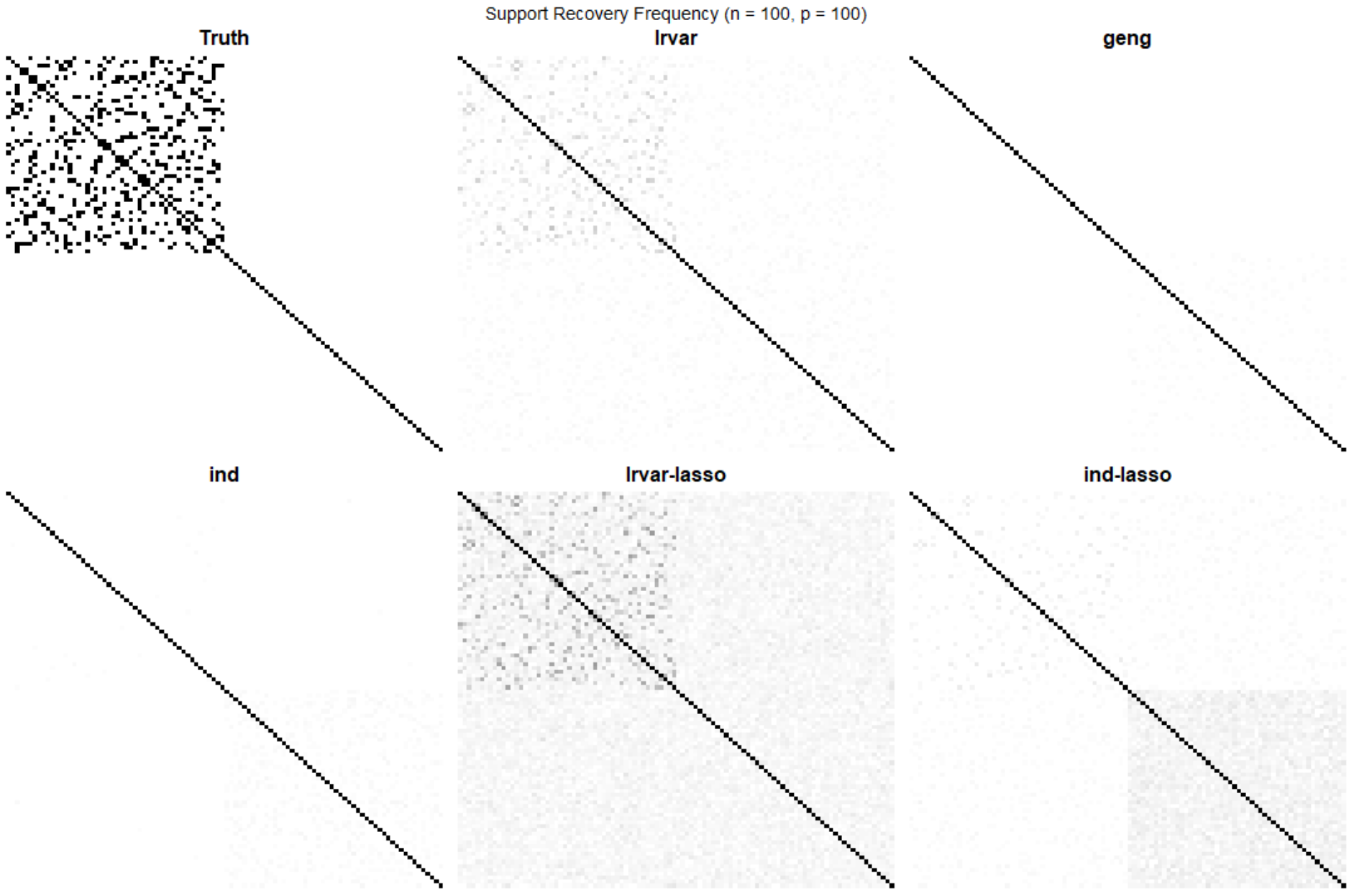}
    
    \caption{Heatmaps of the frequency of the nonzero entries identified for each entry of the Model 1 covariance matrix (when $p = 100$ and $n=100$) out of 100 replications. Black indicates 100 nonzeros identified out of 100 runs; white indicates 0/100.}
    \label{fig:model1_n100}
\end{figure}

\begin{figure}[htbp]
    \centering
    
        \includegraphics[width=\textwidth]{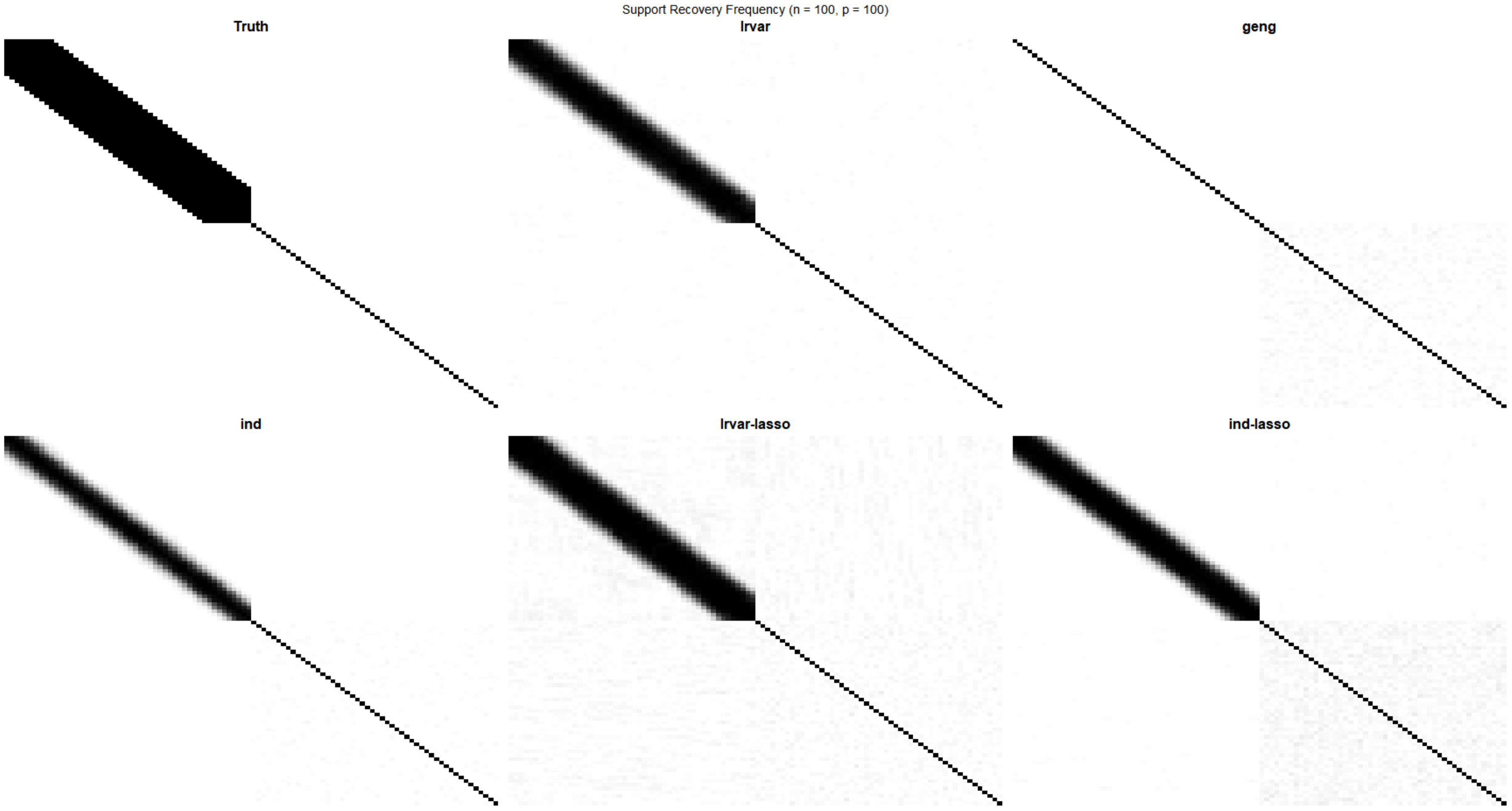}
    
    \caption{Heatmaps of the frequency of the nonzero entries identified for each entry of the Model 2 covariance matrix (when $p = 100$ and $n=100$) out of 100 replications. Black indicates 100 nonzeros identified out of 100 runs; white indicates 0/100.}
    \label{fig:model2_n100}
\end{figure}

Heatmaps of the nonzero entries identified over 100 replications for $p = 100$ and $n = 100$ are presented in Figures \ref{fig:model1_n100} and \ref{fig:model2_n100}. When the dimension $p$ is comparable to or larger than the sample size $n$, the universal thresholding method of \citet{2008Bickel} and the adaptive thresholding method of \citet{2011Cai} become less sensitive to small covariance entries. As a result, many weak but potentially informative signals are removed. In contrast, the proposed method is more effective at preserving these small entries.

\section{Real Data Analysis}\label{sec5}
Additional empirical results are presented using the dataset of \citet{2001Khan} and stock return data from \citet{2020Pelger}.

\subsection{Gene-expression covariance analysis}
We apply the proposed long-run thresholding estimator to the SRBC microarray
dataset of \citet{2001Khan} and compare the sparsity patterns produced by
different thresholding estimators, although the observations are not naturally time ordered. The SRBC dataset was previously analyzed by \citet{2009Rothman} using universal thresholding and by \citet{2011Cai} using adaptive thresholding. To ensure comparability, we follow the same procedures as those in \citet{2009Rothman} and \citet{2011Cai}.

The SRBC dataset contains 63 training tissue samples, each with 2308 gene expression measurements. The original dataset contained 6567 genes and was reduced to 2308 genes after an initial filtering step; see \citet{2001Khan}. The samples belong to four tumor types: 23 EWS, 8 BL-NHL, 12 NB, and 20 RMS. Following \citet{2009Rothman}, we rank the genes according to their discriminative power using the $F$-statistic \[
F = \frac{1}{k-1} \sum_{m=1}^{k} n_m (\bar{x}_m - \bar{x})^2 \Bigg/ \bigg( \frac{1}{n-k} \sum_{m=1}^{k} (n_m - 1) \hat{\sigma}_m^2 \bigg),
\]
where $n = 63$ denotes the sample size and $k = 4$ the number of tumor classes. For $1 \leq m \leq 4$, let $n_m$ denote the sample size of class $m$, and let $\bar{x}_m$ and $\hat{\sigma}_m^2$ denote the sample mean and sample variance of class $m$, respectively. The overall sample mean is denoted by $\bar{x}$.  

Based on the $F$-statistic, we select the top 40 genes and the bottom 160 genes. The top 40 genes are further ordered following \citet{2009Rothman}. Using these 200 genes, we evaluate the performance of the three estimators $\boldsymbol{\widehat{\Sigma}}_y^*$, $\boldsymbol{\widehat{\Sigma}}_y^g$, and $\boldsymbol{\widehat{\Sigma}}_y^c$. The tuning parameter $\delta$ is selected by five-fold cross-validation. Specifically, the 63 samples are partitioned into five groups of approximately equal size. To preserve class balance, the proportions of the four tumor types are kept nearly equal across the groups.  

\begin{figure}[htbp]
    \centering
    
        \includegraphics[width=\textwidth]{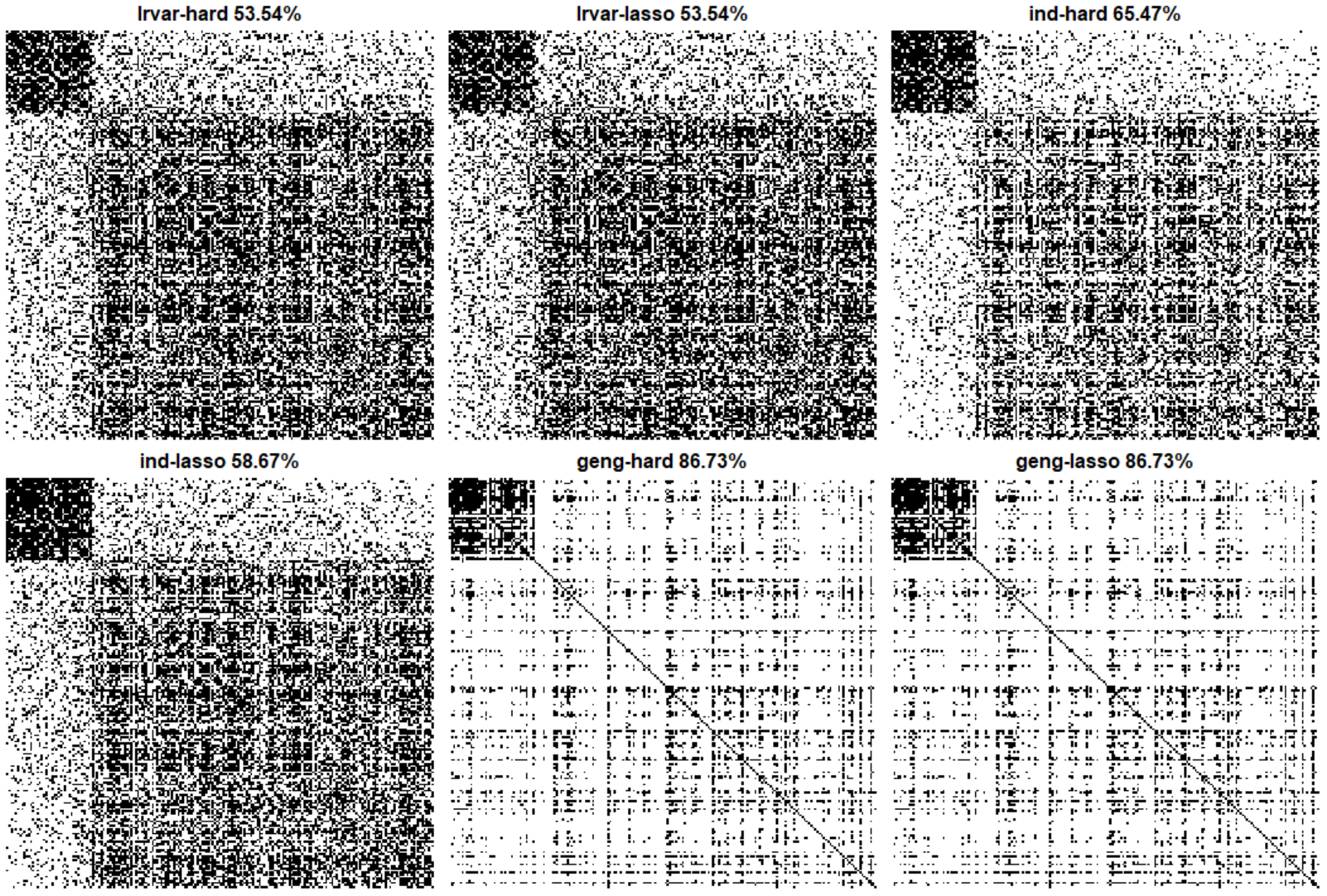}
    
    \caption{ Heatmaps of the estimated supports}
    \label{fig:real data}
\end{figure}

Figure \ref{fig:real data} presents heatmaps of the estimated supports under different thresholding methods. For the universal thresholding estimator $\boldsymbol{\widehat{\Sigma}}_y^g$, the results are similar under both hard thresholding and adaptive lasso thresholding, yielding highly sparse estimates with 86.73$\%$ zero off-diagonal entries. The adaptive thresholding estimator \(\boldsymbol{\widehat\Sigma}_y^c\) produces less sparse estimates,
with 65.47$\%$ zero off-diagonal entries under hard thresholding and 58.67$\%$ under adaptive
lasso thresholding. The proposed estimator \(\boldsymbol{\widehat\Sigma}_y^*\) is less sparse, with
53.54$\%$ zero off-diagonal entries under both hard thresholding and adaptive lasso
thresholding, indicating a more balanced recovery of the covariance structure.  

These findings are consistent with the simulation results. Although the variables are selected using the $F$-statistic, correlations among the bottom-ranked variables may still exist. The resulting sparsity pattern suggests that the universal thresholding estimator may be overly sparse. Since the true support is unknown in real data, this comparison should be interpreted as an empirical illustration of the sparsity patterns produced by different methods.   

In contrast, the proposed method based on long-run variance produces a less aggressive sparsity pattern and may retain more potentially relevant covariance entries.

\subsection{Portfolio optimization}
\citet{1952Markowitz} defined the mean–variance optimal portfolio (MVP) as the solution to the following optimization problem:
\begin{equation}  \label{eq17}
\begin{split}
        &\min_{\boldsymbol{\xi} \in \mathbb{R}^p} \ \boldsymbol{\xi}' \boldsymbol{\Sigma}_y \boldsymbol{\xi}
\quad \\&\text{subject to} \quad \boldsymbol{\xi}' \mathbf{1} = 1, \ \boldsymbol{\xi}' \boldsymbol{\mu} = \gamma_n,
\end{split}
\end{equation}
where \(\mathbf{1}\) denotes a \(p \times 1\) vector of ones, $\boldsymbol{\mu} = E(\boldsymbol y_t)$, and $\gamma_n$ is the target expected return.  

It is well known (see \citet{1959Markowitz} and \citet{2009Cochrane}) that the optimal portfolio is given by
\begin{equation}  \label{eq18}
\boldsymbol{\xi}_n =
\frac{\theta_n - \gamma_n \psi_n}{\phi_n \theta_n - \psi_n^2} \boldsymbol{\Sigma}_y^{-1} \mathbf{1}
+
\frac{\gamma_n \phi_n - \psi_n}{\phi_n \theta_n - \psi_n^2} \boldsymbol{\Sigma}_y^{-1} \boldsymbol{\mu},
\end{equation}
where
\[
\phi_n = \mathbf{1}' \boldsymbol{\Sigma}_y^{-1} \mathbf{1}, \quad
\psi_n = \mathbf{1}' \boldsymbol{\Sigma}_y^{-1} \boldsymbol{\mu}, \quad
\theta_n = \boldsymbol{\mu}' \boldsymbol{\Sigma}_y^{-1} \boldsymbol{\mu}.
\]

The corresponding portfolio variance is
\begin{equation}  \label{eq19}
\boldsymbol{\xi}_n' \boldsymbol{\Sigma}_y \boldsymbol{\xi}_n =
\frac{\theta_n \gamma_n^2 - 2 \psi_n \gamma_n + \phi_n}{\phi_n \theta_n - \psi_n^2}.
\end{equation}

Denote by $\xi_n^g$ the portfolio obtained from $\boldsymbol{\xi}_n$ in (\ref{eq18}) by replacing $\gamma_n$ with $\psi_n/\phi_n$.

The corresponding global minimum variance portfolio (GMVP), without a return constraint, satisfies
\begin{equation}  \label{eq20}
\boldsymbol{\xi}_{ng}' \boldsymbol{\Sigma}_y \boldsymbol{\xi}_{ng} = \phi_n^{-1}.
\end{equation}

This result is obtained by minimizing (\ref{eq19}), or equivalently by setting \(\gamma_n = \psi_n / \phi_n\) in (\ref{eq18}).

Based on the observed data, we construct the covariance estimator $\boldsymbol{\widehat{\Sigma}}_y^*$ as before. We also estimate the mean vector by
\begin{equation}  \label{eq21}
\widehat{\boldsymbol{\mu}}_n = \frac{1}{n} \sum_{t=1}^n \boldsymbol{y}_t.
\end{equation}

To evaluate the out-of-sample portfolio performance of different covariance matrix estimators, we use overnight log excess returns for a balanced panel of S$\&$P 500 stocks from the study of \citet{2020Pelger}, Understanding Systematic Risk: A High-Frequency Approach. The sample period spans January 1, 2004 to December 30, 2005. The dataset contains 332 stocks observed over $T = 504$ trading days. To rank the stocks, we use the mean absolute correlation coefficient, defined by \[\text{Score}_i = \frac{\sum_{j \neq i} |\hat{r}_{ij}|}{p - 1},\] where $p = 332$ is the number of stocks and $\hat{r}_{ij}$ denotes the sample correlation coefficient between stocks $i$ and $j$. Based on this score, we select the top 20 and bottom 80 stocks. Using these 100 stocks, we evaluate out-of-sample portfolio performance using the sample covariance matrix and the four estimators $\boldsymbol{\widehat{\Sigma}}_{\mathrm{hard}}^*$, $\boldsymbol{\widehat{\Sigma}}_{\mathrm{lasso}}^*$, $\boldsymbol{\widehat{\Sigma}}_{LS}$, and $\boldsymbol{\widehat{\Sigma}}_{NLS}$. The stock ranking is used only to construct a fixed empirical universe for comparison and is not intended as part of a real-time trading rule.

At each time $t$, we use the previous $n$ days, from $t-n$ to $t-1$, as the training window to estimate the covariance matrix and compute the portfolio weights. The resulting weight vector $\hat{w}_0$ is then used to calculate portfolio returns over the next 20 trading days. The window is subsequently shifted forward by 20 days, and the procedure is repeated. This corresponds to a monthly rebalancing strategy in which the portfolio is held for 20 trading days before being re-optimized. This procedure yields a sequence of out-of-sample daily portfolio returns over the evaluation period. The annualized realized risk is obtained by multiplying the sample standard deviation of the out-of-sample portfolio returns by \(\sqrt{250}\). We also report
the annualized Sharpe ratio, defined by $\operatorname{SR}=\sqrt{250}\bar r/\widehat{\operatorname{sd}}(r_t)$, where \(r_t\) denotes the daily out-of-sample portfolio excess return, \(\bar r\) is its sample mean, and \(\widehat{\operatorname{sd}}(r_t)\) is its sample standard deviation.

We consider different training-window lengths, with $n = 50,55,\ldots,120$ and $p = 100$. When $n \leq p$, the sample covariance matrix is not invertible, and standard portfolio estimators such as MVP and GMVP cannot be applied directly. For all covariance estimators used in portfolio construction, if the estimated matrix is not positive definite, eigenvalues smaller than \(10^{-6}\) are replaced by \(10^{-6}\) before portfolio optimization.

We compare with two shrinkage estimators. The linear shrinkage estimator
\(\boldsymbol{\widehat\Sigma}^{LS}\) of \citet{2004LedoitWolf} is defined as $\boldsymbol{\widehat\Sigma}^{LS}=(1-\widehat\alpha)\boldsymbol{\widehat{\Sigma}}_y+\widehat\alpha \boldsymbol{\widehat\mu} I_p$, $\boldsymbol{\widehat\mu}=\frac{1}{p}\operatorname{tr}(\boldsymbol{\widehat{\Sigma}}_y)$,
where \(\widehat\alpha\) is the
data-driven shrinkage intensity. The nonlinear shrinkage estimator
\(\boldsymbol{\widehat\Sigma}^{NLS}\) of \citet{2012LedoitWolf} is constructed by preserving
the sample eigenvectors of \(\boldsymbol{\widehat{\Sigma}}_y\) and replacing its sample eigenvalues with nonlinearly shrunk eigenvalues.

\begin{figure}[htbp]
    \centering
    
        \includegraphics[width=\textwidth,height=0.4\textheight]{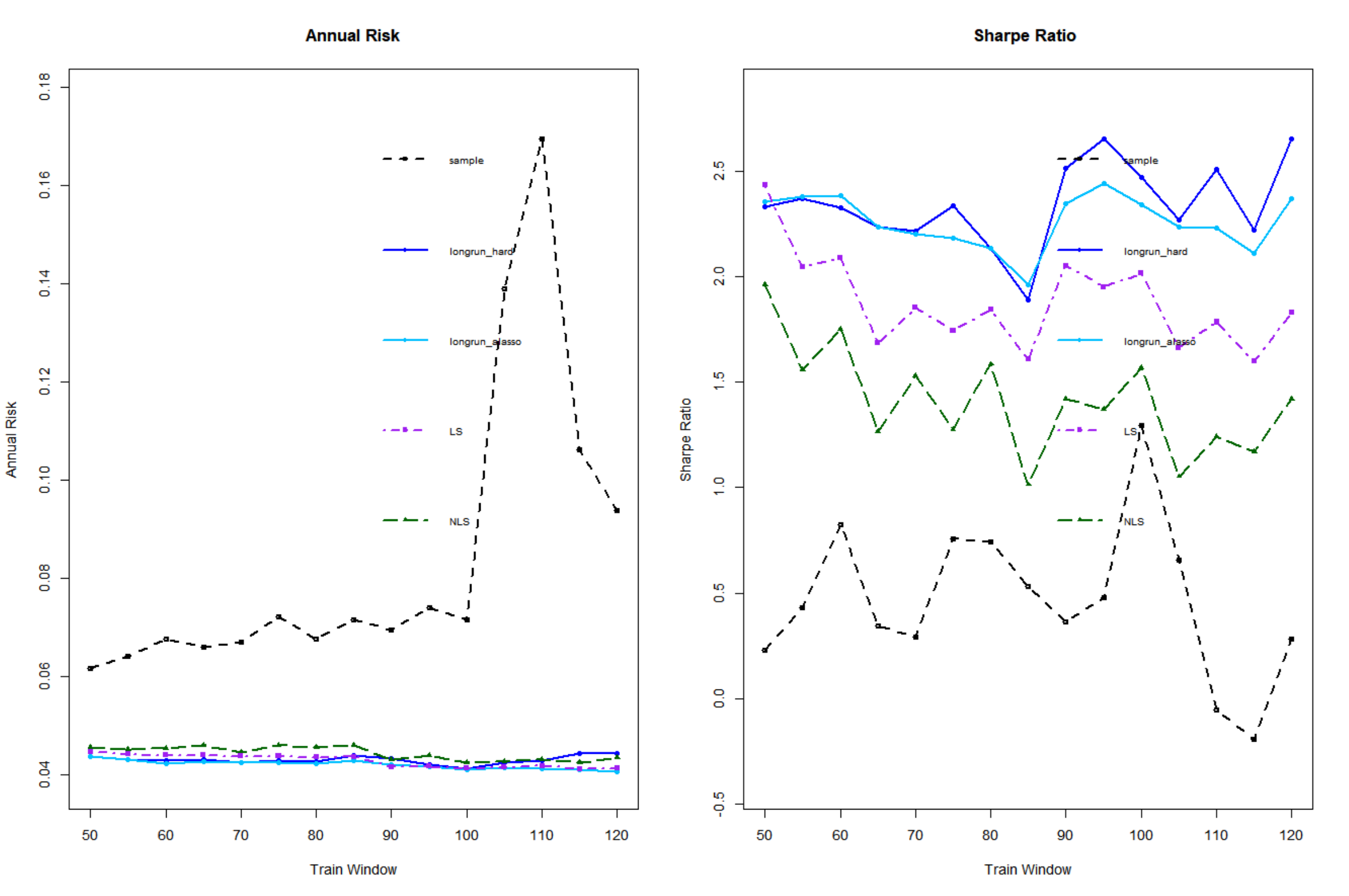}
    
    \caption{ Annualized risk and Sharpe ratio of mean--variance portfolios with a 10$\%$ annual target return.}
    \label{fig:MVP}
\end{figure}

Figure \ref{fig:MVP} reports the annualized risk and Sharpe ratio of the mean--variance portfolios for different training-window lengths. The proposed long-run-variance-based estimators show competitive out-of-sample performance compared with the sample covariance matrix and the shrinkage estimators. 
In particular, the adaptive-lasso version of the proposed estimator generally yields stable risk performance, while the hard-thresholding version remains competitive in terms of the Sharpe ratio.

\begin{figure}[htbp]
    \centering
    
        \includegraphics[width=\textwidth,height=0.4\textheight]{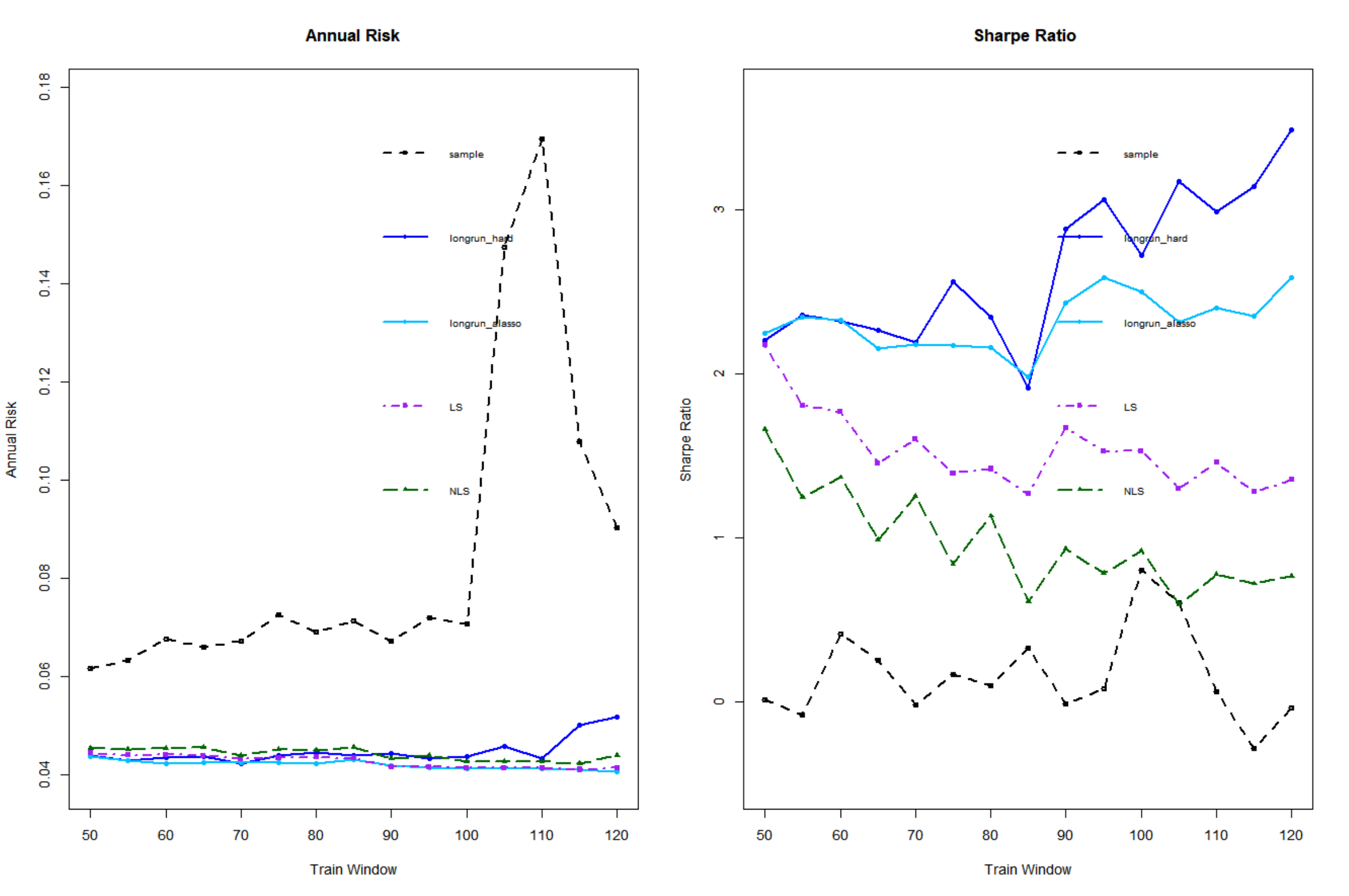}
    
    \caption{ Global minimum variance portfolio annualized risk and Sharpe ratio}
    \label{fig:GMVP}
\end{figure}

Figure \ref{fig:GMVP} presents the performance of global minimum variance portfolios under different covariance estimators. The long-run-variance-based estimator using the adaptive lasso rule achieves the lowest and most stable risk, while the hard-thresholding estimator performs marginally better than LS and NLS for moderate window lengths. In terms of the Sharpe ratio, the long-run-variance-based estimators consistently outperform the alternatives, with hard thresholding showing a slight advantage over lasso. Overall, the long-run-variance-based estimator using the adaptive lasso rule yields the most stable risk reduction, while the hard-thresholding estimator can be competitive in terms of the Sharpe ratio.

\section{Conclusion}\label{sec6}
In this paper, we develop a thresholding approach for covariance matrix estimation in high-dimensional time series. Theoretically, the proposed estimator is shown to achieve the spectral-norm convergence rate and to recover the support consistently under suitable signal-strength conditions. To address the challenges posed by temporal dependence, we incorporate long-run variance into the thresholding framework. Both simulation studies and applications to the SRBC gene-expression dataset and
S$\&$P 500 stock returns illustrate the competitive performance and practical usefulness
of the proposed method. Future work will focus on developing theoretically justified procedures for selecting regularization parameters.

\section{Appendix: Proofs}

We use \(C\) as a generic constant, the value of which may change at different places.

\begin{lemma}  \label{le1}
Let $\rho_1^{-1}=2\gamma_1^{-1}+\gamma_2^{-1}$, $\rho_2^{-1}=\gamma_1^{-1}+\gamma_2^{-1}$. If Assumptions \ref{ass3}(i) and \ref{ass4} hold, we have
\begin{equation}  \label{eq22}
\begin{split}
P(\max_{1\leq i\leq p,1\leq j\leq p}|\hat{\sigma}_{ij}-\sigma_{ij}|>x)&\leq Cp^2n\exp(-Cx^{\rho_1}n^{\rho_1})+Cp^2n\exp(-Cx^{\rho_2/2}n^{\rho_2})\\&+Cp^2\exp(-Cx^2n)+Cp^2\exp(-Cxn) ,
\end{split}
\end{equation}
for any $x>0$, such that $nx\to\infty$.
\end{lemma}

\begin{proof}[\textnormal{\textbf{Proof}}]
\begin{equation}  \label{eq23}
\begin{split}
    \hat{\sigma}_{ij}-\sigma_{ij}&=\frac{1}{n}\sum_{t=1}^n(y_{it}y_{jt}-E(y_{it}y_{jt}))-\frac{1}{n}\sum_{t=1}^n y_{it}\bar{y}_{j}-\frac{1}{n}\sum_{t=1}^n y_{jt}\bar{y}_{i}+\bar{y}_{i}\bar{y}_{j}\\
    &=I_1+I_2+I_3+I_4.
\end{split}
\end{equation}
By Bonferroni inequality,
\begin{equation}  \label{eq24}
    P(\max_{i,j}|\hat{\sigma}_{ij}-\sigma_{ij}|>x)\leq \sum_{i=1}^p\sum_{j=1}^pP(|\hat{\sigma}_{ij}-\sigma_{ij}|>x)\leq Cp^2\sup_{ij}P(|\hat{\sigma}_{ij}-\sigma_{ij}|>x).
\end{equation}
Thus, it suffices to bound $P(|\hat{\sigma}_{ij}-\sigma_{ij}|>x)$ for for \(i,j\in\{1,2,\ldots,p\}\).. By the decomposition in (\ref{eq23}), we only need to bound each term in (\ref{eq23}).

Under Assumption \ref{ass3}(i), we have
\begin{equation}  \label{eq25}
\begin{split}
P(|y_{it}y_{jt}-E(y_{it}y_{jt})|>x)\leq C\exp(-Cx^{\gamma_1/2}).
\end{split}
\end{equation}
By Assumptions \ref{ass3}(i), \ref{ass4} and \citet{2011Merlevede} Theorem 1,
\begin{equation}  \label{eq26}
    P(|I_1|>x)\leq Cn\exp(-Cx^{\rho_1}n^{\rho_1})+C\exp(-Cx^2n),
\end{equation}
where we incorporate the third term in Theorem 1 of \citet{2011Merlevede} into the second one.

Similarly,
\begin{equation}  \label{eq27}
\begin{split}
P(|I_2|>x)&\leq P(|\bar{y}_{i}|>\sqrt{x})+P(|\bar{y}_{j}|>\sqrt{x})\\&\leq Cn\exp(-Cx^{\rho_2/2}n^{\rho_2})+C\exp(-Cxn).
\end{split}
\end{equation}
\begin{equation}  \label{eq28}
    P(|I_3|>x)\leq Cn\exp(-Cx^{\rho_2/2}n^{\rho_2})+C\exp(-Cxn).
\end{equation}
\begin{equation}  \label{eq29}
    P(|I_4|>x)\leq Cn\exp(-Cx^{\rho_2/2}n^{\rho_2})+C\exp(-Cxn).
\end{equation}
Thus, by (\ref{eq26})-(\ref{eq29}),
\begin{equation}  \label{eq30}
\begin{split}
P(|\hat{\sigma}_{ij}-\sigma_{ij}|>x)&\leq Cn\exp(-Cx^{\rho_1}n^{\rho_1})+Cn\exp(-Cx^{\rho_2/2}n^{\rho_2})\\&+C\exp(-Cx^2n)+C\exp(-Cxn). 
\end{split}
\end{equation}
Then follows from the Bonferroni inequality in (\ref{eq24}), this completes the proof.
\end{proof}

\begin{lemma}  \label{le2}
Let $\rho_3^{-1}=4\gamma_1^{-1}+\gamma_2^{-1}$. If Assumptions \ref{ass3}(i), \ref{ass4} and \ref{ass5} hold, we have
\begin{equation}  \label{eq31}
\begin{split}
P(\max_{ij}|\widehat{\theta}_{ij}-\widetilde{\theta}_{ij}|>x)&\leq Cp^2n^2\exp(-C(\frac{nx}{b_n})^{\rho_3})+Cp^2n\exp(-C\frac{nx^2}{b_n^2})\\&+Cp^2n^2\exp(-Cn^{\rho_1}(\frac{x}{b_n})^{\rho_1/2})+Cp^2n\exp(-C\frac{nx}{b_n}).        
\end{split}
\end{equation}
for any $x>0$ such that $nx\to\infty$, where $\widetilde{\theta}_{ij}=\sum_{k=-n+1}^{n-1}\mathcal{K}(\frac{k}{b_n})\widetilde{\Gamma}_{ij}(k)$ with 
\begin{equation}  \label{eq32}
\widetilde{\Gamma}_{ij}(k)=\begin{cases}
    \frac{1}{n}\sum_{t=k+1}^nE[z_{ij,t}-\sigma_{ij}][z_{ij,t-k}-\sigma_{ij}]\quad \text{if $k\geq0$.}\\ \frac{1}{n}\sum_{t=1-k}^nE[z_{ij,t+k}-\sigma_{ij}][z_{ij,t}-\sigma_{ij}]\quad \text{if $k<0$.}   
\end{cases}
\end{equation}
As a result, if $\log p=o\{min(\frac{n}{b_n^2},n^{\rho_3/(2-\rho_3)})\}$ with a bandwidth $b_n^2=o(n)$, $\max_{ij}|\widehat{\theta}_{ij}-\theta_{ij}|\xrightarrow{p}0$ as $n\to \infty$.
\end{lemma}

\begin{proof}[\textnormal{\textbf{Proof}}]
By a similar argument with Lemma \ref{le1},  we only need to bound $P(|\widehat{\theta}_{ij}-\widetilde{\theta}_{ij}|>x)$ for $1\leq i,j\leq p$.

Note that \[\widehat{\theta}_{ij}-\widetilde{\theta}_{ij}=\sum_{k=-n+1}^{n-1}\mathcal{K}(\frac{k}{b_n})[\widehat{\Gamma}_{ij}(k)-\widetilde{\Gamma}_{ij}(k)],\]
where $\widehat{\Gamma}_{ij}(k)$ is defined in (\ref{eq4}) and $\widetilde{\Gamma}_{ij}(k)$ is defined in (\ref{eq32}). For simplicity, we only focus on the case when $k>0$ since the other part is similar. We denote \[a_{ij,t}=z_{ij,t}-\sigma_{ij}.\] After some elementary calculations, we have 
\begin{equation}  \label{eq33}
\begin{split}
& \left|\sum_{k=0}^{n-1} \mathcal{K}\left(\frac{k}{b_n}\right) \frac{1}{n} \sum_{t=k+1}^n\left[z_{i j, t}-\bar{z}_{i j}\right]\left[z_{i j, t-k}-\bar{z}_{i j}\right]-E\left[z_{i j, t}-\sigma_{i j}\right]\left[z_{i j, t-k}-\sigma_{i j}\right]\right| \\&=\left|\sum_{k=0}^{n-1} \mathcal{K}\left(\frac{k}{b_n}\right) \frac{1}{n} \sum_{t=k+1}^n\left[a_{i j, t}-\bar{a}_{i j}\right]\left[a_{i j, t-k}-\bar{a}_{i j}\right]-E\left[a_{i j, t}a_{i j, t-k}\right]\right|
\\&=\left|\sum_{k=0}^{n-1} \mathcal{K}\left(\frac{k}{b_n}\right) \frac{1}{n} \sum_{t=k+1}^n(a_{ij,t}a_{ij,t-k})-E(a_{ij,t}a_{ij,t-k})-\bar{a}_{ij}(a_{ij,t}+a_{ij,t-k}-\bar{a}_{ij})\right|
\\&\leq \left|\sum_{k=0}^{n-1} \mathcal{K}\left(\frac{k}{b_n}\right) \frac{1}{n} \sum_{t=k+1}^n(a_{ij,t}a_{ij,t-k})-E(a_{ij,t}a_{ij,t-k})\right|\\&+\left|\sum_{k=0}^{n-1} \mathcal{K}\left(\frac{k}{b_n}\right) \frac{1}{n} \sum_{t=k+1}^n\bar{a}_{ij}(a_{ij,t}+a_{ij,t-k}-\bar{a}_{ij})\right|
\\&\leq  b_n \max _{0 \leq k \leq n-1}\left|\frac{1}{n} \sum_{t=k+1}^n(a_{i j, t}a_{i j, t-k})-E(a_{i j, t}a_{i j, t-k})\right| \frac{1}{b_n} \sum_{k=0}^{n-1} |\mathcal{K}\left(\frac{k}{b_n}\right)| \\
& +C b_n \max _{0 \leq k \leq n-1}\left|\frac{1}{n} \sum_{t=k+1}^n a_{i j, t}\right|^2 \frac{1}{b_n} \sum_{k=0}^{n-1} |\mathcal{K}\left(\frac{k}{b_n}\right)|,
    \end{split}
\end{equation}
by Assumption \ref{ass5}, $\frac{1}{b_n} \sum_{k=0}^{n-1}| \mathcal{K}\left(\frac{k}{b_n}\right)|<\int|\mathcal{K}(x)|dx<C$. Then
\begin{equation}  \label{eq34}
    \begin{split}
        |\sum_{k=-n+1}^{n-1}\mathcal{K}(\frac{k}{b_n})[\widehat{\Gamma}_{ij}(k)-\widetilde{\Gamma}_{ij}(k)]|
        &\leq C b_n \max _{0 \leq k \leq n-1}\left|\frac{1}{n} \sum_{t=k+1}^n(a_{i j, t}a_{i j, t-k})-E(a_{i j, t}a_{i j, t-k})\right|  \\
& +C b_n \max _{0 \leq k \leq n-1}\left|\frac{1}{n} \sum_{t=k+1}^n a_{i j, t}\right|^2 
\\&=J_1+J_2.
    \end{split}
\end{equation}
Let $\eta_{ij,tk}=a_{ij,t}a_{ij,t-k}$. It follows that \[\sup_{i,j}\sup_{t,k}P(|
\eta_{ij,tk}-E(\eta_{ij,tk})|>x)\leq C\exp(-Cx^{\gamma_1/4}).\]
For $\sum_{t=k+1}^n
\eta_{ij,tk}-E(\eta_{ij,tk})$, we can always split it for three parts: For fixed \(k>0\), define
\[
\eta_{ij,t}^{(k)}
=
a_{ij,t}a_{ij,t-k}
-
E(a_{ij,t}a_{ij,t-k}),
\qquad t=k+1,\ldots,n .
\]
Let
\[
q_k=\max\left\{1,\left\lfloor \frac{k}{2}\right\rfloor\right\},
\qquad
N_k=\left\lceil \frac{n-k}{q_k}\right\rceil .
\]
Partition the index set \(\{k+1,\ldots,n\}\) into consecutive blocks
\[
B_{k,\ell}
=
\left\{
k+(\ell-1)q_k+1,\ldots,
\min(k+\ell q_k,n)
\right\},
\qquad
\ell=1,\ldots,N_k .
\]
We divide these blocks into three groups cyclically. Specifically, define
\[
\mathcal I_{k,r}
=
\bigcup_{\ell\ge0:\,3\ell+r\le N_k}
B_{k,3\ell+r},
\qquad r=1,2,3 .
\]
Then
\[
\sum_{t=k+1}^{n}\eta_{ij,t}^{(k)}
=
\sum_{r=1}^{3}
\sum_{t\in\mathcal I_{k,r}}
\eta_{ij,t}^{(k)} .
\]
Consequently,
\[
P\left(
\left|
\sum_{t=k+1}^{n}\eta_{ij,t}^{(k)}
\right|>nx
\right)
\le
\sum_{r=1}^{3}
P\left(
\left|
\sum_{t\in\mathcal I_{k,r}}
\eta_{ij,t}^{(k)}
\right|>\frac{nx}{3}
\right).
\] For fixed \(k\) and \(r=1,2,3\), define
\[
S_{ij,k}^{(r)}
=
\sum_{t\in I_{k,r}}\eta_{ij,t}^{(k)}
=
\sum_m U_{k,m}^{(r)} .
\]
Since
\[
\sum_{t=k+1}^{n}\eta_{ij,t}^{(k)}
=
\sum_{r=1}^{3}S_{ij,k}^{(r)},
\]
we have
\[
P\left(
\left|
\sum_{t=k+1}^{n}\eta_{ij,t}^{(k)}
\right|
>
Cnx/b_n
\right)
\le
\sum_{r=1}^{3}
P\left(
|S_{ij,k}^{(r)}|
>
Cnx/(3b_n)
\right).
\]
Moreover, for each fixed \(r\), the block-sum sequence
\(\{U_{k,m}^{(r)}\}_{m\ge0}\) is strongly mixing, with mixing coefficients bounded by
\[
\alpha_U^{(k,r)}(h)
\le
\alpha_p\left((3h-1)q_k+1-k\right)
\le
C\exp(-c h^{\gamma_2}).
\]
Therefore, by Theorem of \citet{2011Merlevede},
\begin{equation}  \label{eq35}
\begin{split}
        P(|\sum_{t=k+1}^n
\eta_{ij,tk}-E(\eta_{ij,tk})|>Cnxb_n^{-1})&\leq \sum_{r=1}^{3}P\left(|S_{ij,k}^{(r)}|>Cnx/(3b_n)\right)\\& \leq
CN_k\exp(-C(\frac{nx}{b_n})^{\rho_3})+C\exp(-C\frac{n^2x^2}{N_kb_n^2})
\\&\leq Cn\exp(-C(\frac{nx}{b_n})^{\rho_3})+C\exp(-C\frac{nx^2}{b_n^2}).
\end{split}
\end{equation}
By Bonferroni inequality, 
\begin{equation}  \label{eq36}
    P(J_1>x)\leq CnP(|\sum_{t=k+1}^n
\eta_{ij,tk}-E(\eta_{ij,tk})|>Cnxb_n^{-1})\leq Cn^2\exp(-C(\frac{nx}{b_n})^{\rho_3})+Cn\exp(-C\frac{nx^2}{b_n^2}).
\end{equation}
Similarly,
\begin{equation}  \label{eq37}
    P(J_2>x)\leq CnP(|\sum_{t=k+1}^n
a_{ij,t}|>Cn\sqrt{xb_n^{-1}})\leq Cn^2\exp(-C(\frac{n^2x}{b_n})^{\rho_1/2})+Cn\exp(-C\frac{nx}{b_n}).
\end{equation}
For the second part, if $\log p=o\{\min(\frac{n}{b_n^2},n^{\rho_3/(2-\rho_3)})\}$ with a bandwidth $b_n^2=o(n)$, we have $\max_{ij}|\widehat{\theta}_{ij}-\widetilde{\theta}_{ij}|=o_p(1)$. For the second part, we also consider $k\geq 0$, 
\begin{align*}
    \begin{split}
        |\theta_{ij}-\widetilde{\theta}_{ij}|&=|\sum_{k=0}^{n-1}\frac{n-k}{n}(\mathcal{K}(\frac{k}{b_n})\gamma_{z_{ij}}(k)-\gamma_{z_{ij}}(k))| \\&\leq |\sum_{k=0}^{M}\frac{n-k}{n}(\mathcal{K}(\frac{k}{b_n})\gamma_{z_{ij}}(k)-\gamma_{z_{ij}}(k))|+|\sum_{k=M+1}^{n-1}\frac{n-k}{n}(\mathcal{K}(\frac{k}{b_n})\gamma_{z_{ij}}(k)-\gamma_{z_{ij}}(k))|.
    \end{split}
\end{align*}
For a sufficiently large constant $M$, $|\sum_{k=0}^{M}(\frac{n-k}{n}\mathcal{K}(\frac{k}{b_n})\gamma_{z_{ij}}(k)-\gamma_{z_{ij}}(k))|\to 0$. Since $\mathcal{K}(x)\leq C$, according to Davydov inequality, Assumptions \ref{ass3}(i) and \ref{ass4}, $|\sum_{k=M+1}^{n-1}\frac{n-k}{n}(\mathcal{K}(\frac{k}{b_n})\gamma_{z_{ij}}(k)-\gamma_{z_{ij}}(k))|\to 0$. So we have  $\max_{ij}|\theta_{ij}-\widetilde{\theta}_{ij}|=o(1)$.

By the triangle inequality, this completes the proof.
\end{proof}

\begin{proof}[\textnormal{\textbf{Proof of Theorem 1}}]
    In the event that $E=\{\max_{ij}|\hat{\theta}_{ij}-\theta_{ij}|\leq C_{L}/2 \}$, for simplicity, we let $u_{ij}=\delta\sqrt{\frac{\hat{\theta}_{ij}\log p}{n}}$, where $\delta$ is sufficiently large, note that 
    \begin{equation}  \label{eq38}
        ||\widehat{\boldsymbol{\Sigma}}_y^*-\boldsymbol{\Sigma}_y||_2\leq ||\widehat{\boldsymbol{\Sigma}}_y^*-\boldsymbol{\Sigma}_y^*||_2+||\boldsymbol{\Sigma}_y^*-\boldsymbol{\Sigma}_y||_2,
    \end{equation}
    where $\boldsymbol{\Sigma}_y^*$ is defined similarly as  $\widehat{\boldsymbol{\Sigma}}_y^*$ by thresholding the true covariance matrix. We now consider the first term. For any matrix $\textbf{A}$, we have $||\textbf{A}||_2\leq (||\textbf{A}||_1||\textbf{A}||_\infty)^{\frac{1}{2}}$, thus
\begin{equation}  \label{eq39}
\begin{split}
    ||\widehat{\boldsymbol{\Sigma}}^*_y-\boldsymbol{\Sigma}^*_y||_2&\leq[\max_{1\leq j\leq p}\sum_{i=1}^p|s_{u_{ij}}(\hat{\sigma}_{ij})-s_{u_{ij}}(\sigma_{ij})|]^{\frac{1}{2}} \\&\times [\max_{1\leq i\leq p}\sum_{j=1}^p|s_{u_{ij}}(\hat{\sigma}_{ij})-s_{u_{ij}}(\sigma_{ij})|]^{\frac{1}{2}} \\&=I_5^{\frac{1}{2}}\times I_6^{\frac{1}{2}}.
\end{split}
\end{equation}
    By the thresholding function conditions (1) and (2),
\begin{equation}  \label{eq40}
\begin{split}
        I_5&\leq \max_{1\leq j\leq p}\sum_{i=1}^p|\hat{\sigma}_{ij}-\sigma_{ij}|I(|\hat{\sigma}_{ij}|\geq u_{ij},|\sigma_{ij}|\geq u_{ij})
        + \max_{1\leq j\leq p}\sum_{i=1}^p|\hat{\sigma}_{ij}|I(|\hat{\sigma}_{ij}|\geq u_{ij},|\sigma_{ij}|< u_{ij})\\&+\max_{1\leq j\leq p}\sum_{i=1}^p|\sigma_{ij}|I(|\hat{\sigma}_{ij}|< u_{ij},|\sigma_{ij}|\geq u_{ij})\\&+\max_{1\leq j\leq p}\sum_{i=1}^p(|s_{u_{ij}}(\hat{\sigma}_{ij})-\hat{\sigma}_{ij}|+|s_{u_{ij}}(\sigma_{ij})-\sigma_{ij}|)I(|\hat{\sigma}_{ij}|\geq u_{ij},|\sigma_{ij}|\geq u_{ij})
        \\&=I_{5,1}+I_{5,2}+I_{5,3}+I_{5,4}.
\end{split} 
\end{equation}
Let $Z=\max_{1\leq i,j\leq p}|\hat{\sigma}_{ij}-\sigma_{ij}|$, by Lemma \ref{le1}, $Z=O_p(\sqrt{\log p/n})$, by Lemma \ref{le2},
\begin{equation}  \label{eq41}
\begin{split}
    I_{5,1}&\leq \max_{1\leq i,j\leq p}|\hat{\sigma}_{ij}-\sigma_{ij}|\times\max_j\sum_{i=1}^pI(|\hat{\sigma}_{ij}|\geq u_{ij},|\sigma_{ij}|\geq u_{ij})
    \\&\leq Z\times \max_j\sum_{i=1}^p|\sigma_{ij}|^\iota u_{ij}^{-\iota}\leq Cs_1(\frac{\log p}{n})^{\frac{1-\iota}{2}}.
\end{split}
\end{equation}
For $I_{5,2}$, By the triangle inequality and Lemma \ref{le1}, \ref{le2},
\begin{equation}  \label{eq42}
\begin{split}
    I_{5,2}&\leq \max_j\sum_{i=1}^p|\hat{\sigma}_{ij}-\sigma_{ij}|I(|\hat{\sigma}_{ij}|\geq u_{ij},|\sigma_{ij}|< u_{ij})+\max_j\sum_{i=1}^p|\sigma_{ij}|I(|\sigma_{ij}|<u_{ij})
    \\&\leq \max_j\sum_{i=1}^p|\hat{\sigma}_{ij}-\sigma_{ij}|I(|\hat{\sigma}_{ij}|\geq u_{ij},|\sigma_{ij}|< u_{ij})+\max_j\sum_{i=1}^p|\sigma_{ij}||\frac{u_{ij}}{\sigma_{ij}}|^{1-\iota}
    \\&=\max_j\sum_{i=1}^p|\hat{\sigma}_{ij}-\sigma_{ij}|I(|\hat{\sigma}_{ij}|\geq u_{ij},|\sigma_{ij}|< u_{ij})+Cs_1(\frac{\log p}{n})^{\frac{1-\iota}{2}}.
\end{split}
\end{equation}
Taking $\alpha\in(0,1)$, by the triangle inequality and Lemma \ref{le1}, \ref{le2}, the first term satisfies
\begin{equation}  \label{eq43}
\begin{split}
    &\max_j\sum_{i=1}^p|\hat{\sigma}_{ij}-\sigma_{ij}|I(|\hat{\sigma}_{ij}|\geq u_{ij},|\sigma_{ij}|< u_{ij})
    \\&\leq\max_j\sum_{i=1}^p|\hat{\sigma}_{ij}-\sigma_{ij}|I(|\hat{\sigma}_{ij}|\geq u_{ij},|\sigma_{ij}|< \alpha u_{ij})
    \\&\quad+\max_j\sum_{i=1}^p|\hat{\sigma}_{ij}-\sigma_{ij}|I(|\hat{\sigma}_{ij}|\geq u_{ij},\alpha u_{ij}\leq|\sigma_{ij}|< u_{ij})
    \\&\leq Z\times\max_j\sum_{i=1}^pI(|\hat{\sigma}_{ij}-\sigma_{ij}|\geq (1-\alpha)u_{ij})+Z\times \max_j\sum_{i=1}^p|\sigma_{ij}|^\iota(\alpha u_{ij})^{-\iota}
    \\&\leq Z\times\max_j\sum_{i=1}^pI(|\hat{\sigma}_{ij}-\sigma_{ij}|\geq (1-\alpha)u_{ij})+Cs_1(\frac{\log p}{n})^{\frac{1-\iota}{2}}.
\end{split}
\end{equation}
On the other hand, for any $\delta_1>1$, for a sufficiently large $\delta$, by Markov inequality and Lemma \ref{le1},
\begin{equation}  \label{eq44}
\begin{split}
    &P(\max_j\sum_{i=1}^pI(|\hat{\sigma}_{ij}-\sigma_{ij}|\geq (1-\alpha)u_{ij})>\delta_1)
    \\&\leq \delta_1^{-1}\sum_{i,j=1}^pP(|\hat{\sigma}_{ij}-\sigma_{ij}|>(1-\alpha)u_{ij})=O(p^{-c}).
\end{split}
\end{equation}
By (\ref{eq42}), (\ref{eq43}) and (\ref{eq44}),
\begin{equation}  \label{eq45}
    I_{5,2}\leq C(\frac{\log p}{n})^{\frac{1}{2}}+Cs_1(\frac{\log p}{n})^{\frac{1-\iota}{2}}=O_p(s_1(\frac{\log p}{n})^{\frac{1-\iota}{2}}).
\end{equation}

For $I_{5,3}$, by the triangle inequality and Lemma \ref{le1}, \ref{le2},
\begin{equation}  \label{eq46}
\begin{split}
    I_{5,3}&\leq \max_{1\leq j\leq p}\sum_{i=1}^p|\hat{\sigma}_{ij}-\sigma_{ij}|I(|\hat{\sigma}_{ij}|< u_{ij},|\sigma_{ij}|\geq u_{ij}) + \max_{1\leq j\leq p}\sum_{i=1}^p|\hat{\sigma}_{ij}|I(|\hat{\sigma}_{ij}|< u_{ij},|\sigma_{ij}|\geq u_{ij}) \\&\leq Z\times \max_j\sum_{i=1}^pI(|\sigma_{ij}|>u_{ij})+\max_{1\leq j\leq p}\sum_{i=1}^p|\sigma_{ij}|^{\iota}u_{ij}^{1-\iota}I(|\hat{\sigma}_{ij}|< u_{ij},|\sigma_{ij}|\geq u_{ij})\\&\leq Z\times \max_j\sum_{i=1}^p|\frac{\sigma_{ij}}{u_{ij}}|^\iota+Cs_1(\frac{\log p}{n})^{\frac{1-\iota}{2}}
    \leq Cs_1(\frac{\log p}{n})^{\frac{1-\iota}{2}}.
\end{split}
\end{equation}

For $I_{5,4}$, applying condition (3), 
\begin{equation}
\begin{split}
    &\max_{1\leq j\leq p}\sum_{i=1}^p(|s_{u_{ij}}(\hat{\sigma}_{ij})-\hat{\sigma}_{ij}|+|s_{u_{ij}}(\sigma_{ij})-\sigma_{ij}|)I(|\hat{\sigma}_{ij}|\geq u_{ij},|\sigma_{ij}|\geq u_{ij}) \\& \leq\max_{1\leq j\leq p}2\sum_{i=1}^p u_{ij}^{\iota}u_{ij}^{1-\iota}I(|\hat{\sigma}_{ij}|\geq u_{ij},|\sigma_{ij}|\geq u_{ij})\\& \leq \max_{1\leq j\leq p}2\sum_{i=1}^p |\sigma_{ij}|^{\iota}u_{ij}^{1-\iota} \leq   Cs_1(\frac{logp}{n})^{\frac{1-\iota}{2}}
\end{split}
\end{equation}

Together with $I_{5,1}$, $I_{5,2}$, $I_{5,3}$ and $I_{5,4}$, we have $I_5=O_p(s_1(\frac{\log p}{n})^{\frac{1-\iota}{2}})$. By a similar argument as $I_5$, we can show that $I_6=O_p(s_1(\frac{\log p}{n})^{\frac{1-\iota}{2}})$.

Now, turn to the second term. Note that
\begin{equation}  \label{eq47}
    ||\boldsymbol{\Sigma}_y^*-\boldsymbol{\Sigma}_y||_2\leq [\max_i\sum_{j=1}^p|s_{u_{ij}}(\sigma_{ij})-\sigma_{ij}|]^{\frac{1}{2}}\times[\max_j\sum_{i=1}^p|s_{u_{ij}}(\sigma_{ij})-\sigma_{ij}|]^{\frac{1}{2}}.
\end{equation}
Applying the thresholding function conditions (2), (3) and  Assumption \ref{ass1}, 
\begin{align*}
    &\max_i\sum_{j=1}^p|s_{u_{ij}}(\sigma_{ij})-\sigma_{ij}|\leq \max_i\sum_{j=1}^p |\sigma_{ij}|I(|\sigma_{ij}|\leq u_{ij})+\max_i\sum_{j=1}^p u_{ij}I(|\sigma_{ij}|> u_{ij})\\& =\max_i\sum_{j=1}^p |\sigma_{ij}|^{\iota}|\sigma_{ij}|^{1-\iota}I(|\sigma_{ij}|\leq u_{ij})+\max_i\sum_{j=1}^p u_{ij}^{\iota}u_{ij}^{1-\iota}I(|\sigma_{ij}|> u_{ij}) \\& \leq \max_i\sum_{j=1}^p|\sigma_{ij}|^\iota u_{ij}^{1-\iota}=Cs_1(\frac{logp}{n})^{\frac{1-\iota}{2}}
\end{align*}
Similarly, \[\max_j\sum_{i=1}^p|s_{u_{ij}}(\sigma_{ij})-\sigma_{ij}|\leq   Cs_1(\frac{\log p}{n})^{\frac{1-\iota}{2}},\] so we have 
\begin{equation}  \label{eq48}
    ||\boldsymbol{\Sigma}_y^*-\boldsymbol{\Sigma}_y||_2\leq Cs_1(\frac{\log p}{n})^{\frac{1-\iota}{2}}.
\end{equation}
According to $I_5$, $I_6$ , (\ref{eq48}) and Lemma \ref{le2}, the proof is completed.
\end{proof}

\begin{proof}[\textnormal{\textbf{Proof of Theorem 2}}]
According to Lemma \ref{le1} and \ref{le2}, when $\log p=o\{\min(\frac{n}{b_n^2},n^{\rho_3/(2-\rho_3)})\}$ with a bandwidth $b_n^2=o(n)$, for a sufficiently large $\delta$, define the event $E=\{\max_{ij}|\hat{\theta}_{ij}-\theta_{ij}|\leq C_{L}/2 \}$, we have 
\begin{equation} \label{eq2.2}
    \begin{split}
        P(\max_{ij}|\frac{\hat{\sigma}_{ij}-\sigma_{ij}}{\hat{\theta}^{\frac{1}{2}}_{ij}}|>\delta\sqrt{\frac{logp}{n}})&\leq P((\max_{ij}|\frac{\hat{\sigma}_{ij}-\sigma_{ij}}{\hat\theta_{ij}^\frac{1}{2}}|>\delta\sqrt{\frac{logp}{n}})\cap E)+P(E^c)\\&\leq P(\max_{ij}|\hat{\sigma}_{ij}-\sigma_{ij}|>C\delta\sqrt{\frac{logp}{n}})+o(1)=o(1)
    \end{split}
\end{equation}
We first consider the case $\sigma_{ij}=0$. According to (\ref{eq2.2}), when $\sigma_{ij}=0$, we have \[P(\max_{(i,j)\notin\Psi}|\hat{\sigma}_{ij}|<\lambda_{ij})\to1.\] 
Then we consider the case $\sigma_{ij}\neq0$. According to (\ref{eq2.2}), we have \[|\hat{\sigma}_{ij}|\geq |\sigma_{ij}|-|\hat{\sigma}_{ij}-\sigma_{ij}|\geq |\sigma_{ij}|-\lambda_{ij}.\]
Since $|\sigma_{ij}|>(2+\eta)\delta\sqrt{\frac{\theta_{ij}\log p}{n}}$, and according to Lemma \ref{le2}, we have $|\sigma_{ij}|\geq (2+\eta/2)\lambda_{ij}$ with probability tending to 1. So we uniformly have $|\hat{\sigma}_{ij}|>\lambda_{ij}$ in $(i,j)\in\Psi$ with probability tending to 1. 
    
Thus the result follows. This completes the proof.
\end{proof}

\begin{lemma}  \label{le5}
Suppose that $\{\boldsymbol{y}_t\}$ is a Gaussian $\mathrm{VAR}(1)$ process
\[
    \boldsymbol{y}_t=\boldsymbol{\Phi} y_{t-1}+\boldsymbol{\varepsilon}_t,\qquad
    \boldsymbol{\Phi}=\operatorname{diag}(\phi_1,\ldots,\phi_p),
\]
where $\max_{1\le j\le p}|\phi_j|\le r_0<1$ for some constant $r_0\in(0,1)$.
Assume that Assumptions 3--5 hold and that
\[
    \xi\log n\le \log p
    =o\left\{\min\left(\frac{n}{b_n^2},\,n^{\rho_3/(2-\rho_3)}\right)\right\},
    \qquad b_n^2=o(n),
\]
for some $\xi>0$, where $\rho_3^{-1}=4\gamma_1^{-1}+\gamma_2^{-1}$.
Assume moreover that
\[
    \max_{1\le i\le p}\#\{j\ne i:\sigma_{ij}\ne 0\}\le s_1-1,
    \qquad 8\leq s_1=O((\log p)^\gamma)
\]
for some $\gamma<1$. For each $i$, define
\[
    B_i=\{j\ne i:\sigma_{ij}=0\}.
\]
Let
\[
    \lambda_{ij}=\tau\sqrt{\frac{\hat\theta_{ij}\log p}{n}},
    \qquad 0<\tau<\sqrt{2}.
\]
Then, for any
\[
    0<\epsilon_0<\frac12\left(1-\frac{\tau^2}{2}\right),
\]
we have
\[
    P\left(
    \min_{1\le i\le p}
    \sum_{j\in B_i}
    I\{|\widehat\sigma_{ij}|\ge \lambda_{ij}\}
    >p^{\epsilon_0}
    \right)\to 1 .
\]
\end{lemma}

\begin{proof}[\textnormal{\textbf{Proof}}]
Fix $i\in\{1,\ldots,p\}$. Since each index is correlated with at most $s_1-1$
other indices, we have
\[
    |B_i|\ge p-s_1.
\]
Choose a constant $\alpha$ such that
\[
    2\epsilon_0+\frac{\tau^2}{2}<\alpha<1.
\]
We first construct a large subset of $B_i$ whose coordinates are pairwise uncorrelated.
Starting from $B_i$, choose one index and put it into $J_i$. Then delete this selected
index together with all indices correlated with it. Repeat this procedure until no index
remains. Since each selected index is correlated with at most $s_1-1$ other indices, each
step removes at most $s_1$ indices. Therefore
\[
    |J_i|
    \ge
    \frac{|B_i|}{s_1}
    \ge
    \frac{p-s_1}{s_1}.
\]
Because $s_1=O((\log p)^\gamma)$ with $\gamma<1$, for any fixed $\alpha<1$,
\[
    \frac{p-s_1}{s_1}\ge p^\alpha
\]
for all sufficiently large $p$. Hence
\[
    |J_i|\ge p^\alpha.
\]

By construction, the coordinates indexed by $J_i$ are pairwise uncorrelated. Since the
process is Gaussian, pairwise uncorrelatedness implies independence. Moreover, because
$j\in B_i$ implies $\sigma_{ij}=0$, the coordinate process $\{y_{jt}\}$ is independent of
$\{y_{it}\}$. Indeed, for the diagonal $\mathrm{VAR}(1)$ process,
\[
    \operatorname{Cov}(y_{it},y_{j,t+h})
    =
    \begin{cases}
    \phi_j^h\sigma_{ij}, & h\ge0,\\
    \phi_i^{-h}\sigma_{ij}, & h<0,
    \end{cases}
\]
and hence all cross-lag covariances vanish when $\sigma_{ij}=0$. Therefore, for
$j,k\in J_i$, $j\ne k$, the coordinate processes $\{y_{jt}\}$ and $\{y_{kt}\}$ are
independent, and they are also independent of $\{y_{it}\}$.

Partition $J_i$ into
\[
    N_p=\lfloor p^{2\epsilon_0}\rfloor
\]
disjoint subsets
\[
    H_{i,1},H_{i,2},\ldots,H_{i,N_p},
\]
such that
\[
    |H_{i,m}|\asymp p^{\alpha-2\epsilon_0},
    \qquad
    1\le m\le N_p.
\]

We next control the conditional variance of
\[
    \sum_{t=1}^n y_{it}y_{jt},
    \qquad j\in J_i.
\]
Since the data can be centered by subtracting their sample means, we assume without loss of generality that $\boldsymbol{\mu}=0$. By rescaling each coordinate, the event under consideration is invariant. Hence we may assume $\sigma_{ii}=1$ for $1\leq i \leq p$. For $j\in J_i$, conditionally on $\{y_{it}:1\le t\le n\}$, $\sum_{t=1}^n y_{it}y_{jt}$ is Gaussian with conditional mean zero and conditional variance
\[
    V_{ij,n}
    =
    \sum_{t=1}^n y_{it}^2
    +
    2\sum_{h=1}^{n-1}(\phi_j)^h
    \sum_{t=1}^{n-h}y_{it}y_{i,t+h}.
\]
Define
\[
    \hat\gamma_i(h)
    =
    \frac{1}{n-h}
    \sum_{t=1}^{n-h}y_{it}y_{i,t+h},
    \qquad
    \gamma_i(h)=E(y_{it}y_{i,t+h}).
\]
Then
\[
    \frac{V_{ij,n}}{n}
    =
    \hat\gamma_i(0)
    +
    2\sum_{h=1}^{n-1}\frac{n-h}{n}\phi_j^h\hat\gamma_i(h).
\]
For $j\in B_i$, since $\sigma_{ij}=0$ and $\sigma_{jj}=1$, the long-run variance is
\[
    \theta_{ij}
    =
    \gamma_i(0)
    +
    2\sum_{h=1}^{n-1}\frac{n-h}{n}\phi_j^h\gamma_i(h).
\]
Hence
\[
\begin{aligned}
    \left|
    \frac{V_{ij,n}}{n}-\theta_{ij}
    \right|
    &\le
    |\widehat\gamma_i(0)-\gamma_i(0)|  \\
    &\quad
    +2\sum_{h=1}^{n-1}\frac{n-h}{n}|\phi_j|^h
    |\widehat\gamma_i(h)-\gamma_i(h)|.
\end{aligned}
\]
Therefore,
\[
\begin{aligned}
    \max_{i,j}
    \left|
    \frac{V_{ij,n}}{n}-\theta_{ij}
    \right|
    &\le
    \max_i|\widehat\gamma_i(0)-\gamma_i(0)| \\
    &\quad
    +2\max_i\max_{1\le h\le n-1}
    \frac{1}{n}|\sum_{t=1}^{n-h}
    \{y_{it}y_{i,t+h}-E(y_{it}y_{i,t+h})\}|
    \sum_{h=1}^{n-1}r_0^h.
\end{aligned}
\]
Since $\sum_{h=1}^{n-1}r_0^h\le r_0/(1-r_0)<\infty$, by Assumptions \ref{ass3}(i) and \ref{ass4}, together with the Bernstein-type inequality of
\citet{2011Merlevede}, for some constants $C,c>0$,
\[
    P\left(
    \left|
    \sum_{t=1}^{n-h}
    \{y_{it}y_{i,t+h}-E(y_{it}y_{i,t+h})\}
    \right|
    >
    nx
    \right)
    \le
    Cn\exp\{-c(nx)^{\rho_1}\}
    +
    C\exp\{-cnx^2\}.
\]
By the union bound over $1\le i\le p$ and $1\le h\le n-1$,
\[
\begin{aligned}
    &P\left(
    \max_i\max_{1\le h\le n-1}
    \left|
    \frac1n
    \sum_{t=1}^{n-h}
    \{y_{it}y_{i,t+h}-E(y_{it}y_{i,t+h})\}
    \right|
    >
    x
    \right)  \\
    &\qquad
    \le
    Cnp
    \left[
    n\exp\{-c(nx)^{\rho_1}\}
    +
    \exp\{-cnx^2\}
    \right].
\end{aligned}
\]
Consequently,
\[
    \max_{i,j}
    \left|
    \frac{V_{ij,n}}{n}-\theta_{ij}
    \right|
    =
    o_p(1).
\]
Moreover, by Lemma \ref{le2},
\[
    \max_{i,j}|\hat\theta_{ij}-\theta_{ij}|=o_p(1).
\]
According to Lemma \ref{le1}, $\max_{i,j}|\bar y_i\bar y_j|=O_p\left(\frac{\log p}{n}\right)$.
Choose $\varepsilon>0$ sufficiently small so that
\[
    \alpha-2\epsilon_0-\frac{(\tau+2\varepsilon)^2}{2}>0.
\]
For $j\in J_i$, define
\[
    C_{ij}
    =
    \left\{
    \frac{|\sum_{t=1}^n y_{it}y_{jt}|}{\sqrt{V_{ij,n}}}
    \ge
    (\tau+\varepsilon)\sqrt{\log p}
    \right\}.
\]
On the event
\[
    \max_{i,j}|\widehat\theta_{ij}-\theta_{ij}|=o(1),
    \qquad
    \max_{i,j}
    \left|
    \frac{V_{ij,n}}{n}-\theta_{ij}
    \right|=o(1),
\]
and using Assumption \ref{ass3}(ii), we have uniformly in $i,j$,
\[
    \sqrt{\frac{V_{ij,n}}{n}}
    =
    (1+o(1))\sqrt{\theta_{ij}},
    \qquad
    \sqrt{\hat\theta_{ij}}
    =
    (1+o(1))\sqrt{\theta_{ij}}.
\]
Therefore, for all sufficiently large $n$, the event $C_{ij}$ implies
\[
    \frac1n|\sum_{t=1}^n y_{it}y_{jt}|
    \ge
    (\tau+\varepsilon/2)
    \sqrt{\frac{\hat\theta_{ij}\log p}{n}}.
\]
Since
\[
    \max_{ij}|\bar y_i\bar y_j|
    =O_p(\frac{\log p}{n})=
    o_p\left(\sqrt{\frac{\log p}{n}}\right),
\]
we further obtain
\[
    |\hat\sigma_{ij}|
    =
    \left|
    \frac1n\sum_{t=1}^n y_{it}y_{jt}
    -
    \bar y_i\bar y_j
    \right|
    \ge
    \tau
    \sqrt{\frac{\hat\theta_{ij}\log p}{n}}
    =
    \lambda_{ij}
\]
with probability tending to one. Hence it is enough to show that, with probability
tending to one, every block $H_{i,m}$ contains at least one index $j$ for which $C_{ij}$
occurs.

Now fix $i$ and $m$. Conditionally on $\{y_{it}:1\le t\le n\}$, the random variables
\[
    \left\{
    \sum_{t=1}^n y_{it}y_{jt}:j\in H_{i,m}
    \right\}
\]
are independent Gaussian random variables with variances $V_{ij,n}$. Thus
\[
    P(C_{ij}\mid \{y_{it}\}_{t=1}^n)
    =
    P\{|Z|\ge(\tau+\varepsilon)\sqrt{\log p}\},
    \qquad Z\sim N(0,1).
\]
Defined the event $E=\{\max_{ij}|\hat\theta_{ij}-\theta_{ij}|\leq \varepsilon\}$, according to Lemma \ref{le2}, we have $P(E^c)=o(1)$. Using the standard lower bound for the normal tail,
\[
    P\{|Z|\ge x\}
    \ge
    c x^{-1}\exp(-x^2/2),
    \qquad x>1,
\]
we obtain, for all sufficiently large $p$,
\[
    P(C_{ij}\mid \{y_{it}\}_{t=1}^n)
    \ge
    c p^{-(\tau+2\varepsilon)^2/2}.
\]
Therefore,
\[
\begin{aligned}
    P\left(
    \sum_{j\in H_{i,m}}I(C_{ij})=0
    \,\middle|\,
    \{y_{it}\}_{t=1}^n
    \right)
    &\le
    \left(
    1-cp^{-(\tau+2\varepsilon)^2/2}
    \right)^{|H_{i,m}|}  \\
    &\le
    \exp\left\{
    -c|H_{i,m}|p^{-(\tau+2\varepsilon)^2/2}
    \right\}.
\end{aligned}
\]
Since
\[
    |H_{i,m}|\asymp p^{\alpha-2\epsilon_0},
\]
we have
\[
    P\left(
    \sum_{j\in H_{i,m}}I(C_{ij})=0
    \right)
    \le
    \exp\left\{
    -c p^{\alpha-2\epsilon_0-(\tau+2\varepsilon)^2/2}
    \right\}.
\]
By the choice of $\varepsilon$,
\[
    \alpha-2\epsilon_0-\frac{(\tau+2\varepsilon)^2}{2}>0,
\]
so
\[
    P\left(
    \sum_{j\in H_{i,m}}I(C_{ij})=0
    \right)
    =
    o(p^{-1-2\epsilon_0})
\]
uniformly over $i$ and $m$.

Finally, applying the union bound over $1\le i\le p$ and
$1\le m\le N_p$, we get
\[
\begin{aligned}
    P\left(
    \exists\, i,m:
    \sum_{j\in H_{i,m}}I(C_{ij})=0
    \right)
    &\le
    pN_p\,o(p^{-1-2\epsilon_0})  \\
    &=
    p\lfloor p^{2\epsilon_0}\rfloor o(p^{-1-2\epsilon_0})  \\
    &=o(1).
\end{aligned}
\]
Thus, with probability tending to one, for every $i$ and every block $H_{i,m}$, there
exists at least one $j\in H_{i,m}$ such that $C_{ij}$ occurs, and hence $|\widehat\sigma_{ij}|\ge \lambda_{ij}$. Therefore,
\[\sum_{j\in B_i}I\{|\widehat\sigma_{ij}|\ge\lambda_{ij}\}\ge N_p=\lfloor p^{2\epsilon_0}\rfloor>p^{\epsilon_0}\]
for all sufficiently large $p$, uniformly over $1\le i\le p$. Hence
\[P\left(\min_{1\le i\le p}\sum_{j\in B_i}I\{|\hat\sigma_{ij}|\ge\lambda_{ij}\}>p^{\epsilon_0}\right)\to 1.\]
This completes the proof.
\end{proof}

\begin{proof}[\textnormal{\textbf{Proof of Theorem 3}}]
Let
\[
\theta_0:=\frac{1+\rho^2}{1-\rho^2},
\qquad
\delta_0:=\frac12\sqrt{\theta_0}.
\]
Choose a fixed constant \(\rho\in(0,1)\) sufficiently close to one, and then choose \(c_a>0\) such that
\[
(2+\eta)\delta_* <c_a<\delta_0/2-M,
\]
where \(M>0\) is the constant in the Gaussian maximal deviation bound
\begin{equation}\label{80}
    P\!\left(
\max_{1\le i\neq j\le s_0}
|\hat\sigma_{ij}-\sigma_{ij}|
\le
M\sqrt{\frac{\log p}{n}}
\right)\to1,
\end{equation}
according to Lemma \ref{le1}.
By the choice \(c_a>(2+\eta)\delta_*\), the signal strength condition of Theorem \ref{the2}
holds for the proposed estimator on the first block, and hence the proposed estimator
still satisfies \(P(\mathrm{TPR}=1)\to1\) and \(P(\mathrm{FPR}=0)\to1\).

We now show that the universal thresholding estimator fails.

\textbf{Case 1: \(\delta\le\delta_0\).}
Fix the last row \(i=p\), and let
\[
J_p=\{s_0+1,\ldots,p-1\}.
\]
For all \(j\in J_p\), we have \(\sigma_{pj}=0\), since the second block is diagonal.
Moreover,
\[
\theta_{pj}=\theta_0,
\qquad j\in J_p.
\]
By Lemma \ref{le2},
\[
\max_{j\in J_p}|\hat\theta_{pj}-\theta_0|=o_p(1),
\]
so with probability tending to one,
\[
\hat\theta_{pj}\ge \frac12\theta_0,
\qquad j\in J_p.
\]
Hence
\[
\lambda_g=\delta\sqrt{\frac{\log p}{n}}
\le
\delta_0\sqrt{\frac{\log p}{n}}
=
\frac12\sqrt{\theta_0\frac{\log p}{n}}
\le
\frac{1}{\sqrt2}\sqrt{\hat\theta_{pj}\frac{\log p}{n}},
\qquad j\in J_p,
\]
with probability tending to one.
Now apply Lemma \ref{le5} with \(\tau=1/\sqrt2\). Since \(|J_p|\ge p/2\), Lemma \ref{le5} implies that,
with probability tending to one, at least one zero entry in the \(p\)-th row exceeds the
universal threshold \(\lambda_g\). Therefore
\[
P\bigl(\mathrm{FPR}(\widehat{\boldsymbol{\Sigma}}_y^{\,g})>0\bigr)\to1.
\]

\textbf{Case 2: \(\delta>\delta_0\).}
For the first block, \(1\le i\neq j\le s_0\), we have
\[
\sigma_{ij}=a_n=c_a\sqrt{\frac{\log p}{n}}.
\]
By (\ref{80}), with probability tending to one,
\[
|\hat\sigma_{ij}|
\le
(c_a+M)\sqrt{\frac{\log p}{n}},
\qquad
1\le i\neq j\le s_0.
\]
Since \(\delta>\delta_0\) and \(c_a+M<\delta_0\), it follows that
\[
\lambda_g=\delta\sqrt{\frac{\log p}{n}}
>
(c_a+M)\sqrt{\frac{\log p}{n}}
\]
for all sufficiently large \(n\). Hence all truly nonzero off-diagonal entries in the first
block are thresholded to zero, and therefore
\[
P\bigl(\mathrm{TPR}(\widehat{\boldsymbol{\Sigma}}_y^{g})<1\bigr)\to1.
\]

Combining the two cases yields
\[
\inf_{\delta>0}\sup_{\boldsymbol{\Sigma}_y\in U_{\iota}(s_1)}
P\Bigl(
\mathrm{TPR}(\widehat{\boldsymbol{\Sigma}}_y^{g})<1
\ \text{or}\
\mathrm{FPR}(\widehat{\boldsymbol{\Sigma}}_y^{g})>0
\Bigr)\to1.
\]
This completes the proof.
\end{proof}

\begin{proof}[\textnormal{\textbf{Proof of Theorem 4}}]
We use the same construction and the same constants \(\rho\), \(c_a\), \(\theta_0\), and
\(\delta_0\) as in the proof of Theorem \ref{the3}. We now consider the adaptive thresholding estimator of \citet{2011Cai}.

\textbf{Case 1: \(\delta\le \delta_0\).}
Fix \(i=p\) and let
\[
J_p=\{s_0+1,\ldots,p-1\}.
\]
For each \(j\in J_p\), we have \(\sigma_{pj}=0\). Since the second block is diagonal with
common marginal variance \(1\),
\[
\theta_{pj}=\frac{1+\rho^2}{1-\rho^2}=\theta_0,
\qquad
\theta^c_{pj}=\mathrm{Var}(y_{pt}y_{jt})=1.
\]
Moreover, it follows from the same argument as in Lemma \ref{le1} that
\[
\max_{j\in J_p}|\hat\theta^c_{pj}-1|=o_p(1),
\qquad
\max_{j\in J_p}|\hat\theta_{pj}-\theta_0|=o_p(1).
\]
Hence, with probability tending to one,
\[
\hat\theta^c_{pj}\le 2,
\qquad
\hat\theta_{pj}\ge \frac12\theta_0,
\qquad j\in J_p.
\]
Therefore
\[
\lambda^c_{pj}
=
\delta\sqrt{\hat\theta^c_{pj}\frac{\log p}{n}}
\le
\delta_0\sqrt{2\frac{\log p}{n}}
\le
\sqrt{\hat\theta_{pj}\frac{\log p}{n}},
\]
after choosing \(\rho\) sufficiently close to one so that \(\theta_0\) is sufficiently large
relative to \(1\).
Now apply Lemma \ref{le5} with \(\tau=1\). Then, with probability tending to one, at least one
zero entry in the \(p\)-th row is retained by the adaptive thresholding estimator. Hence
\[
P\bigl(\mathrm{FPR}(\boldsymbol{\widehat\Sigma}_y^{c})>0\bigr)\to1.
\]

\textbf{Case 2: \(\delta>\delta_0\).}
For \(1\le i\neq j\le s_0\), we have
\[
\sigma_{ij}=a_n=c_a\sqrt{\frac{\log p}{n}}.
\]
Since the first block is white noise,
\[
\theta^c_{ij}=\mathrm{Var}(y_{it}y_{jt})=1+a_n^2,
\]
and according to Lemma 2 in \citet{2011Cai}, $P(\max_{ij}|\hat{\theta}_{ij}^c-\theta_{ij}^c|>\varepsilon)=O(s_0^{-C})$, so
\[
\lambda^c_{ij}
=
\delta\sqrt{\hat\theta^c_{ij}\frac{\log p}{n}}
\ge
\delta/2\sqrt{\frac{\log p}{n}}
>
(c_a+M)\sqrt{\frac{\log p}{n}}
\]
for all sufficiently large \(n\). By the same maximal deviation argument as in Theorem \ref{the3},
\[
|\hat\sigma_{ij}|
\le
(c_a+M)\sqrt{\frac{\log p}{n}}
<
\lambda^c_{ij},
\qquad
1\le i\neq j\le s_0,
\]
with probability tending to one. Hence all truly nonzero off-diagonal entries in the first
block are thresholded to zero, and thus
\[
P\bigl(\mathrm{TPR}(\widehat{\boldsymbol{\Sigma}}_y^{c})<1\bigr)\to1.
\]

Combining the two cases yields
\[
\inf_{\delta>0}\sup_{\boldsymbol{\Sigma}_y\in U_{\iota}(s_1)}
P\Bigl(
\mathrm{TPR}(\widehat{\boldsymbol{\Sigma}}_y^{c})<1
\ \text{or}\
\mathrm{FPR}(\widehat{\boldsymbol{\Sigma}}_y^{c})>0
\Bigr)\to1.
\]
This completes the proof.
\end{proof}

\end{document}